\documentclass[conference, a4paper, 10pt]{IEEEtran}

\usepackage{times}
\usepackage[latin1]{inputenc}
\usepackage[T1]{fontenc}
\usepackage[english]{babel}
\usepackage{graphicx}
\usepackage{amsmath, amsfonts, dsfont,amssymb, graphicx, array, tabularx, booktabs}
\usepackage{amsthm}
\usepackage{epsf,epsfig}
\usepackage{float}
\usepackage{setspace}
\usepackage[linesnumbered,ruled,vlined]{algorithm2e}
\usepackage{multirow}
\usepackage{url}
\usepackage{color}
\usepackage[nolist,printonlyused]{acronym}      
\usepackage{comment}
\usepackage{cite}
\usepackage{bm}

\usepackage[font=small,labelfont=bf]{caption}

\usepackage{hyperref}
\usepackage{mathtools}

\usepackage{blindtext}

\renewcommand{\arraystretch}{1.2}
\newcommand{\gf}[1]{\textcolor{cyan}{{#1}}}
\newcommand{\mt}[1]{\textcolor{green}{{#1}}}

\newcommand{\ignore}[1]{}

\newtheorem{thm}{Theorem} 

\newtheorem{cor}{Corollary}
\newtheorem{lem}{Lemma}
\newtheorem{prop}{Proposition}

\newtheorem{observ}{Observation}

\newcommand{\mx}[1]{\mathbf{#1}}
\newcommand{\bs}[1]{\boldsymbol{#1}}

\newcommand{\rev}[1]{\textcolor{black}{{#1}}}
\newcommand{\qq}[1]{\textcolor{black}{{#1}}}

\definecolor{amber}{rgb}{1.0, 0.49, 0.0}
\definecolor{ao}{rgb}{0.0, 0.5, 0.0}

\def\R2#1{\textcolor{black}{#1}}
\def\R3#1{\textcolor{black}{#1}}

\SetKwInOut{Parameter}{Parameter}

\begin{document}
\date{}

\title{MU-MIMO Receiver Design and Performance Analysis in Time-Varying Rayleigh Fading}
\title{On the Achievable SINR in MU-MIMO Systems Operating in Time-Varying Rayleigh Fading}
\author{
G\'{a}bor Fodor$^{\star\dag}$, Sebastian Fodor$^\flat$, Mikl\'{o}s Telek$^{\ddag\sharp}$ \\
\small $^\star$Ericsson Research, Stockholm, Sweden. E-mail: \texttt{Gabor.Fodor@ericsson.com}\\
\small $^\dag$KTH Royal Institute of Technology, Stockholm, Sweden. E-mail: \texttt{gaborf@kth.se}\\
\small $^\flat$Stockholm University, Stockholm, Sweden. E-mail: \texttt{sebbifodor@fastmail.com}\\
\small $^\ddag$Budapest University of Technology and Economics, Budapest, Hungary. E-mail: \texttt{telek@hit.bme.hu}\\
\small $^\sharp$MTA-BME Information Systems Research Group, Budapest, Hungary. E-mail: \texttt{telek@hit.bme.hu}
\vspace{-1.7\baselineskip}
}
\maketitle
\pagestyle{plain}

\begin{acronym}
  \acro{2G}{Second Generation}
  \acro{3G}{3$^\text{rd}$~Generation}
  \acro{3GPP}{3$^\text{rd}$~Generation Partnership Project}
  \acro{4G}{4$^\text{th}$~Generation}
  \acro{5G}{5$^\text{th}$~Generation}
  \acro{AA}{Antenna Array}
  \acro{AC}{Admission Control}
  \acro{AD}{Attack-Decay}
  \acro{ADSL}{Asymmetric Digital Subscriber Line}
	\acro{AHW}{Alternate Hop-and-Wait}
  \acro{AMC}{Adaptive Modulation and Coding}
	\acro{AP}{Access Point}
  \acro{APA}{Adaptive Power Allocation}
  \acro{AR}{autoregressive}
  \acro{ARMA}{Autoregressive Moving Average}
  \acro{ATES}{Adaptive Throughput-based Efficiency-Satisfaction Trade-Off}
  \acro{AWGN}{additive white Gaussian noise}
  \acro{BB}{Branch and Bound}
  \acro{BD}{Block Diagonalization}
  \acro{BER}{bit error rate}
  \acro{BF}{Best Fit}
  \acro{BLER}{BLock Error Rate}
  \acro{BPC}{Binary power control}
  \acro{BPSK}{binary phase-shift keying}
  \acro{BPA}{Best \ac{PDPR} Algorithm}
  \acro{BRA}{Balanced Random Allocation}
  \acro{BS}{base station}
  \acro{CAP}{Combinatorial Allocation Problem}
  \acro{CAPEX}{Capital Expenditure}
  \acro{CBF}{Coordinated Beamforming}
  \acro{CBR}{Constant Bit Rate}
  \acro{CBS}{Class Based Scheduling}
  \acro{CC}{Congestion Control}
  \acro{CDF}{Cumulative Distribution Function}
  \acro{CDMA}{Code-Division Multiple Access}
  \acro{CL}{Closed Loop}
  \acro{CLPC}{Closed Loop Power Control}
  \acro{CNR}{Channel-to-Noise Ratio}
  \acro{CPA}{Cellular Protection Algorithm}
  \acro{CPICH}{Common Pilot Channel}
  \acro{CoMP}{Coordinated Multi-Point}
  \acro{CQI}{Channel Quality Indicator}
  \acro{CRM}{Constrained Rate Maximization}
	\acro{CRN}{Cognitive Radio Network}
  \acro{CS}{Coordinated Scheduling}
  \acro{CSI}{channel state information}
  \acro{CSIR}{channel state information at the receiver}
  \acro{CSIT}{channel state information at the transmitter}
  \acro{CUE}{cellular user equipment}
  \acro{D2D}{device-to-device}
  \acro{DCA}{Dynamic Channel Allocation}
  \acro{DE}{Differential Evolution}
  \acro{DFT}{Discrete Fourier Transform}
  \acro{DIST}{Distance}
  \acro{DL}{downlink}
  \acro{DMA}{Double Moving Average}
	\acro{DMRS}{Demodulation Reference Signal}
  \acro{D2DM}{D2D Mode}
  \acro{DMS}{D2D Mode Selection}
  \acro{DPC}{Dirty Paper Coding}
  \acro{DRA}{Dynamic Resource Assignment}
  \acro{DSA}{Dynamic Spectrum Access}
  \acro{DSM}{Delay-based Satisfaction Maximization}
  \acro{ECC}{Electronic Communications Committee}
  \acro{EFLC}{Error Feedback Based Load Control}
  \acro{EI}{Efficiency Indicator}
  \acro{eNB}{Evolved Node B}
  \acro{EPA}{Equal Power Allocation}
  \acro{EPC}{Evolved Packet Core}
  \acro{EPS}{Evolved Packet System}
  \acro{E-UTRAN}{Evolved Universal Terrestrial Radio Access Network}
  \acro{ES}{Exhaustive Search}
  \acro{FDD}{frequency division duplexing}
  \acro{FDM}{Frequency Division Multiplexing}
  \acro{FER}{Frame Erasure Rate}
  \acro{FF}{Fast Fading}
  \acro{FSB}{Fixed Switched Beamforming}
  \acro{FST}{Fixed SNR Target}
  \acro{FTP}{File Transfer Protocol}
  \acro{GA}{Genetic Algorithm}
  \acro{GBR}{Guaranteed Bit Rate}
  \acro{GLR}{Gain to Leakage Ratio}
  \acro{GOS}{Generated Orthogonal Sequence}
  \acro{GPL}{GNU General Public License}
  \acro{GRP}{Grouping}
  \acro{HARQ}{Hybrid Automatic Repeat Request}
  \acro{HMS}{Harmonic Mode Selection}
  \acro{HOL}{Head Of Line}
  \acro{HSDPA}{High-Speed Downlink Packet Access}
  \acro{HSPA}{High Speed Packet Access}
  \acro{HTTP}{HyperText Transfer Protocol}
  \acro{ICMP}{Internet Control Message Protocol}
  \acro{ICI}{Intercell Interference}
  \acro{ID}{Identification}
  \acro{IETF}{Internet Engineering Task Force}
  \acro{ILP}{Integer Linear Program}
  \acro{JRAPAP}{Joint RB Assignment and Power Allocation Problem}
  \acro{UID}{Unique Identification}
  \acro{IID}{Independent and Identically Distributed}
  \acro{IIR}{Infinite Impulse Response}
  \acro{ILP}{Integer Linear Problem}
  \acro{IMT}{International Mobile Telecommunications}
  \acro{INV}{Inverted Norm-based Grouping}
	\acro{IoT}{Internet of Things}
  \acro{IP}{Internet Protocol}
  \acro{IPv6}{Internet Protocol Version 6}
  \acro{ISD}{Inter-Site Distance}
  \acro{ISI}{Inter Symbol Interference}
  \acro{ITU}{International Telecommunication Union}
  \acro{JOAS}{Joint Opportunistic Assignment and Scheduling}
  \acro{JOS}{Joint Opportunistic Scheduling}
  \acro{JP}{Joint Processing}
	\acro{JS}{Jump-Stay}
    \acro{KF}{Kalman filter}
  \acro{KKT}{Karush-Kuhn-Tucker}
  \acro{L3}{Layer-3}
  \acro{LAC}{Link Admission Control}
  \acro{LA}{Link Adaptation}
  \acro{LC}{Load Control}
  \acro{LOS}{Line of Sight}
  \acro{LP}{Linear Programming}
  \acro{LS}{least squares}
  \acro{LTE}{Long Term Evolution}
  \acro{LTE-A}{LTE-Advanced}
  \acro{LTE-Advanced}{Long Term Evolution Advanced}
  \acro{M2M}{Machine-to-Machine}
  \acro{MAC}{Medium Access Control}
  \acro{MANET}{Mobile Ad hoc Network}
  \acro{MC}{Modular Clock}
  \acro{MCS}{Modulation and Coding Scheme}
  \acro{MDB}{Measured Delay Based}
  \acro{MDI}{Minimum D2D Interference}
  \acro{MF}{Matched Filter}
  \acro{MG}{Maximum Gain}
  \acro{MH}{Multi-Hop}
  \acro{MIMO}{multiple input multiple output}
  \acro{MINLP}{Mixed Integer Nonlinear Programming}
  \acro{MIP}{Mixed Integer Programming}
  \acro{MISO}{Multiple Input Single Output}
  \acro{ML}{maximum likelihood}
  \acro{MLWDF}{Modified Largest Weighted Delay First}
  \acro{MME}{Mobility Management Entity}
  \acro{MMSE}{minimum mean squared error}
  \acro{MOS}{Mean Opinion Score}
  \acro{MPF}{Multicarrier Proportional Fair}
  \acro{MRA}{Maximum Rate Allocation}
  \acro{MR}{Maximum Rate}
  \acro{MRC}{maximum ratio combining}
  \acro{MRT}{Maximum Ratio Transmission}
  \acro{MRUS}{Maximum Rate with User Satisfaction}
  \acro{MS}{mobile station}
  \acro{MSE}{mean squared error}
  \acro{MSI}{Multi-Stream Interference}
  \acro{MTC}{Machine-Type Communication}
  \acro{MTSI}{Multimedia Telephony Services over IMS}
  \acro{MTSM}{Modified Throughput-based Satisfaction Maximization}
  \acro{MU-MIMO}{multiuser multiple input multiple output}
  \acro{MU}{multi-user}
  \acro{NAS}{Non-Access Stratum}
  \acro{NB}{Node B}
  \acro{NE}{Nash equilibrium}
  \acro{NCL}{Neighbor Cell List}
  \acro{NLP}{Nonlinear Programming}
  \acro{NLOS}{Non-Line of Sight}
  \acro{NMSE}{Normalized Mean Square Error}
  \acro{NORM}{Normalized Projection-based Grouping}
  \acro{NP}{Non-Polynomial Time}
  \acro{NRT}{Non-Real Time}
  \acro{NSPS}{National Security and Public Safety Services}
  \acro{O2I}{Outdoor to Indoor}
  \acro{OFDMA}{orthogonal frequency division multiple access}
  \acro{OFDM}{orthogonal frequency division multiplexing}
  \acro{OFPC}{Open Loop with Fractional Path Loss Compensation}
	\acro{O2I}{Outdoor-to-Indoor}
  \acro{OL}{Open Loop}
  \acro{OLPC}{Open-Loop Power Control}
  \acro{OL-PC}{Open-Loop Power Control}
  \acro{OPEX}{Operational Expenditure}
  \acro{ORB}{Orthogonal Random Beamforming}
  \acro{JO-PF}{Joint Opportunistic Proportional Fair}
  \acro{OSI}{Open Systems Interconnection}
  \acro{PAIR}{D2D Pair Gain-based Grouping}
  \acro{PAPR}{Peak-to-Average Power Ratio}
  \acro{P2P}{Peer-to-Peer}
  \acro{PC}{Power Control}
  \acro{PCI}{Physical Cell ID}
  \acro{PDF}{Probability Density Function}
  \acro{PDPR}{pilot-to-data power ratio}
  \acro{PER}{Packet Error Rate}
  \acro{PF}{Proportional Fair}
  \acro{P-GW}{Packet Data Network Gateway}
  \acro{PL}{Pathloss}
  \acro{PPR}{pilot power ratio}
  \acro{PRB}{physical resource block}
  \acro{PROJ}{Projection-based Grouping}
  \acro{ProSe}{Proximity Services}
  \acro{PS}{Packet Scheduling}
  \acro{PSAM}{pilot symbol assisted modulation}
  \acro{PSO}{Particle Swarm Optimization}
  \acro{PZF}{Projected Zero-Forcing}
  \acro{QAM}{Quadrature Amplitude Modulation}
  \acro{QoS}{Quality of Service}
  \acro{QPSK}{Quadri-Phase Shift Keying}
  \acro{RAISES}{Reallocation-based Assignment for Improved Spectral Efficiency and Satisfaction}
  \acro{RAN}{Radio Access Network}
  \acro{RA}{Resource Allocation}
  \acro{RAT}{Radio Access Technology}
  \acro{RATE}{Rate-based}
  \acro{RB}{resource block}
  \acro{RBG}{Resource Block Group}
  \acro{REF}{Reference Grouping}
  \acro{RLC}{Radio Link Control}
  \acro{RM}{Rate Maximization}
  \acro{RNC}{Radio Network Controller}
  \acro{RND}{Random Grouping}
  \acro{RRA}{Radio Resource Allocation}
  \acro{RRM}{Radio Resource Management}
  \acro{RSCP}{Received Signal Code Power}
  \acro{RSRP}{Reference Signal Receive Power}
  \acro{RSRQ}{Reference Signal Receive Quality}
  \acro{RR}{Round Robin}
  \acro{RRC}{Radio Resource Control}
  \acro{RSSI}{Received Signal Strength Indicator}
  \acro{RT}{Real Time}
  \acro{RU}{Resource Unit}
  \acro{RUNE}{RUdimentary Network Emulator}
  \acro{RV}{Random Variable}
  \acro{SAC}{Session Admission Control}
  \acro{SCM}{Spatial Channel Model}
  \acro{SC-FDMA}{Single Carrier - Frequency Division Multiple Access}
  \acro{SD}{Soft Dropping}
  \acro{S-D}{Source-Destination}
  \acro{SDPC}{Soft Dropping Power Control}
  \acro{SDMA}{Space-Division Multiple Access}
  \acro{SER}{Symbol Error Rate}
  \acro{SES}{Simple Exponential Smoothing}
  \acro{S-GW}{Serving Gateway}
  \acro{SINR}{signal-to-interference-plus-noise ratio}
  \acro{SI}{Satisfaction Indicator}
  \acro{SIP}{Session Initiation Protocol}
  \acro{SISO}{single input single output}
  \acro{SIMO}{Single Input Multiple Output}
  \acro{SIR}{signal-to-interference ratio}
  \acro{SLNR}{Signal-to-Leakage-plus-Noise Ratio}
  \acro{SMA}{Simple Moving Average}
  \acro{SNR}{signal-to-noise ratio}
  \acro{SORA}{Satisfaction Oriented Resource Allocation}
  \acro{SORA-NRT}{Satisfaction-Oriented Resource Allocation for Non-Real Time Services}
  \acro{SORA-RT}{Satisfaction-Oriented Resource Allocation for Real Time Services}
  \acro{SPF}{Single-Carrier Proportional Fair}
  \acro{SRA}{Sequential Removal Algorithm}
  \acro{SRS}{Sounding Reference Signal}
  \acro{SU-MIMO}{single-user multiple input multiple output}
  \acro{SU}{Single-User}
  \acro{SVD}{Singular Value Decomposition}
  \acro{TCP}{Transmission Control Protocol}
  \acro{TDD}{time division duplexing}
  \acro{TDMA}{Time Division Multiple Access}
  \acro{TETRA}{Terrestrial Trunked Radio}
  \acro{TP}{Transmit Power}
  \acro{TPC}{Transmit Power Control}
  \acro{TTI}{Transmission Time Interval}
  \acro{TTR}{Time-To-Rendezvous}
  \acro{TSM}{Throughput-based Satisfaction Maximization}
  \acro{TU}{Typical Urban}
  \acro{UE}{User Equipment}
  \acro{UEPS}{Urgency and Efficiency-based Packet Scheduling}
  \acro{UL}{uplink}
  \acro{UMTS}{Universal Mobile Telecommunications System}
  \acro{URI}{Uniform Resource Identifier}
  \acro{URM}{Unconstrained Rate Maximization}
  \acro{UT}{user terminal}
  \acro{VR}{Virtual Resource}
  \acro{VoIP}{Voice over IP}
  \acro{WAN}{Wireless Access Network}
  \acro{WCDMA}{Wideband Code Division Multiple Access}
  \acro{WF}{Water-filling}
  \acro{WiMAX}{Worldwide Interoperability for Microwave Access}
  \acro{WINNER}{Wireless World Initiative New Radio}
  \acro{WLAN}{Wireless Local Area Network}
  \acro{WMPF}{Weighted Multicarrier Proportional Fair}
  \acro{WPF}{Weighted Proportional Fair}
  \acro{WSN}{Wireless Sensor Network}
  \acro{WWW}{World Wide Web}
  \acro{XIXO}{(Single or Multiple) Input (Single or Multiple) Output}
  \acro{ZF}{zero-forcing}
  \acro{ZMCSCG}{Zero Mean Circularly Symmetric Complex Gaussian}
\end{acronym}

\begin{abstract}
Minimizing the symbol error in the uplink of multi-user multiple input multiple output systems is important,
because the symbol error affects the achieved signal-to-interference-plus-noise ratio (SINR) and thereby
the spectral efficiency of the system. Despite the vast literature available on minimum mean
squared error (MMSE) receivers,
previously proposed receivers for
block fading channels do not minimize the symbol error in time-varying Rayleigh fading
channels. Specifically, we show that the true MMSE receiver structure does not only depend on the
statistics of the CSI error, but also on the autocorrelation coefficient of the time-variant channel.
It turns out that calculating the average SINR when using the proposed receiver is highly non-trivial.
In this paper, we employ a random matrix theoretical approach, which allows us to derive \qq{a quasi-closed
form for the average SINR,}
which allows to obtain analytical exact results that give valuable insights into
how the SINR depends on the number of antennas, employed pilot and data power
and the covariance of the time-varying channel.
We benchmark the performance of the proposed receiver against recently proposed receivers
and find that
\rev{the proposed MMSE receiver achieves higher SINR
than the previously proposed ones, and this benefit increases with increasing autoregressive coefficient.}
\end{abstract}

\section{Introduction}
\label{Sec:Intro}
The wireless channels in the uplink of \ac{MU-MIMO} systems can often be advantageously modelled as
\ac{AR} processes, because \ac{AR} channel models capture the time-varying (aging) nature of the channels and
facilitate channel estimation and prediction \cite{Yan:01, Zhang:07B, Lehmann:08, Abeida:10, Hijazi:10, GH:12,Truong:13,
Kong:2015, Chiu:15, Kashyap:17, Kim:20, Yuan:20, Fodor:2021}.
These papers have shown that
exploiting the autoregressive structure of the time-varying Rayleigh fading channel
improves the performance of both \ac{SISO} and \ac{MIMO}
channel estimators and receivers.
The basic rationale for these papers is that in a Rayleigh fading environment, based on the associated Jakes
process, an \ac{AR} model can be built, which allows one to employ Kalman filters for
estimating and predicting the channel state.
Specifically, papers \cite{Zhang:07B}, and \cite{Abeida:10, Hijazi:10, GH:12} consider \ac{SISO} systems
and exploit the memoryful property of the \ac{AR} process for joint channel estimation, equalization and data detection.

Some early works on multiple-antenna receiver design and performance analysis are reported in \cite{Yan:01} and \cite{Lehmann:08}.
The optimal array receiver algorithm for \ac{BPSK} signals is designed in \cite{Yan:01}, while reference \cite{Lehmann:08}
is concerned with the blind estimation and detection of space-time coded symbols transmitted over time-varying
Rayleigh fading channels. More recently, in the context of massive \ac{MU-MIMO} systems, \cite{Truong:13,
Kong:2015, Chiu:15, Kashyap:17, Kim:20, Yuan:20, Fodor:2021} addressed the problem of channel aging and derived
channel estimation, prediction and multi-user receiver algorithms that operate in an \ac{AR} Rayleigh-fading
environment and use Kalman filters or machine learning algorithms for channel prediction.


\setlength{\tabcolsep}{2pt} 
\renewcommand{\arraystretch}{1} 
{\footnotesize
\begin{table*}[ht!]
\centering
\caption{Overview of Related Literature}
\vspace{1mm}
\label{tab:tab2}
\footnotesize
\begin{tabular}{
|p{0.15\textwidth}|@{}
>{\centering}p{0.1\textwidth}|
>{\centering}p{0.18\textwidth}|
>{\centering}p{0.15\textwidth}|
>{\centering}p{0.12\textwidth}|
p{0.22\textwidth}|}
\hline
\hline
\textbf{~~Reference} & \textbf{UL or DL} & \textbf{Channel model and channel estimation} & \textbf{Perf. Indicator}
& 
\textbf{Is asymptotic random matrix theory (RMT) used ?}
& \textbf{~~Comment}
\\
\hline
\hline
Couillet et al., \cite{Couillet:2011} & MIMO MAC & block fading, channel estimation (CE) out of scope (OoS)
& rate region, rate maximization & Yes &  receiver design OoS   \\
\hline
Hanlen et al., \cite{Hanlen:2012} & UL/DL & block fading with correlated MIMO channels, perfect CSI at receiver & capacity & Yes & receiver design OoS  \\
\hline
Couillet et al., \cite{Couillet:2012} & UL/DL & block fading, CE is OoS & capacity and sum-rate & Yes &receiver design OoS  \\
\hline
Wen et al., \cite{Wen:2013} & MIMO MAC & block fading, non-Gaussian, CE is OoS & ergodic mutual information & Yes & receiver design OoS  \\
\hline
Hoydis et al., \cite{Hoydis:2013} & UL/DL & block fading, MMSE CE & achievable rate
& Yes & regularized MMSE receiver that takes into account the estimated channels of all users \\
\hline
Truong et al., \cite{Truong:13} & UL/DL & AR(p), AR(1), MMSE, channel prediction & average SINR, achievable rate & Yes & MRC receiver (not AR-aware) \\
\hline
Kong et al., \cite{Kong:2015} & UL/DL & similar to that in \cite{Truong:13} & UL/DL average rate & Yes & MRC and ZF receivers (not AR-aware) \\
\hline
Papazafeiropoulos et al., \cite{Papa:18} & UL & AR(1), MMSE estimation
& average SINR, outage probability & Yes & MRC receiver, UL caching \\
\hline
\rev{Bj\"{o}rnson et al., \cite{Bjornson:18}} & UL/DL & block fading, multicell MMSE CE & average SINR, spectral efficiency & Yes & multicell MMSE receiver \\
\hline
\rev{Boukhedini et al., \cite{Boukhedimi:18}} & UL/DL & block fading, multicell MMSE CE & average SINR, spectral efficiency & Yes & multicell MMSE receiver \\
\hline
\rev{Sanguinetti et al., \cite{Sanguinetti:19}} &  UL/DL & block fading, multicell MMSE CE & average SINR, spectral efficiency & Yes & multicell MMSE receiver \\
\hline
Yuan et al., \cite{Yuan:20} & UL/DL & AR(1), ML-based prediction
& channel estimation/prediction quality (MSE) & No & receiver design OoS (focus on channel estimation/prediction) \\
\hline
Abrardo et al., \cite{Abrardo:19} & UL & block fading, LS CE
& MSE and SINR & Yes & MMSE receiver for block fading channels is derived; takes into account the estimated channels of all users \\
\hline
Kim et al., \cite{Kim:20} & UL/DL & 3GPP spatial channel model, ML-based and Kalman filter-based prediction, mobility prediction
& channel estimation/prediction quality (MSE) & No & receiver design OoS (focus on channel estimation/prediction)\\
\hline
Fodor et al., \cite{Fodor:2021} & UL & AR(1), Kalman filter-based channel estimation & MSE of received data symbols & No & regularized (AR-aware) MMSE receiver, regularization is based on covariance matrices, interference is treated as nosie \\
\hline
\rev{Chopra and Murthy \cite{Chopra:21}} & UL/DL & AR($p$), Kalman filter-based and data assisted channel estimation
& MSE of the channel estimation and received data symbols, and achievable rate
& Yes & AR-aware MMSE receiver that utilizes data-aided channel tracking \\
\hline
Present paper & UL & AR(1), Kalman filter-based channel estimation & average SINR, average rate
& Yes & new MMSE receiver, whose structure takes into account the AR parameters and estimated channels of all users  \\
\hline
\hline
\end{tabular}
\vspace{-2mm}
\end{table*}
}
\setlength{\tabcolsep}{6pt} 
\renewcommand{\arraystretch}{1} 

A closely related line of research, in block fading environments, applies results from random matrix
theory to establish the deterministic equivalent of the random wireless system in order to calculate
the \ac{SINR} in the uplink and downlink of \rev{\ac{MU-MIMO} systems \cite{Couillet:2011, Hanlen:2012, Couillet:2012, Wen:2013,
Hoydis:2013, Kammoun:2014, Muller:2015, Papa:18, Bjornson:18, Boukhedimi:18, Abrardo:19, Sanguinetti:19}.}
\rev{In particular, in papers \cite{Bjornson:18, Boukhedimi:18, Sanguinetti:19} it was shown that
the capacity of multicell \ac{MU-MIMO} networks grows indefinitely as the number of antennas
tends to infinity, if appropriate multicell \ac{MMSE} processing is used.}

Generalizing the \ac{DL} precoding and \ac{UL} receiver
structures and associated deterministic equivalent \ac{SINR} results developed in these papers
to \ac{AR} time-varying environments and channel aging is not trivial, because of the basic assumption on independent
channel realizations at subsequent time instances.
In contrast, \rev{papers \cite{Truong:13, Kong:2015, Papa:18, Kim:20, Yuan:20}} treat \ac{AR} channel evolution
and use random matrix theory to derive the deterministic equivalent and thereby the \ac{SINR} for the \ac{UL} and \ac{DL}
of \ac{MU-MIMO} systems. However, these papers do not develop a \ac{MU-MIMO} receiver that aims to minimize
the \ac{MSE} of the received data symbols.
\rev{More recently, paper \cite{Chopra:21} developed a data-aided
\ac{MSE}-optimal channel tracking scheme and associated \ac{MMSE} estimator of the data symbols in the
presence of channel aging, that is when the channel changes between the channel estimation time instance
and the time instance when the channel is used for data transmission.}
%

In our recent work \cite{Fodor:2021}, we developed a new \ac{MMSE}
receiver that treats interference as noise and uses an \ac{AR} model for its performance analysis (see Table \ref{tab:tab2}).
The important conclusion in \cite{Fodor:2021} is that not only the channel estimation procedure, but the
receiver structure itself should be modified when the fading process is \ac{AR}.

However, it is well-known that treating interference as noise in \ac{MU-MIMO} systems
can severely degrade the performance as compared with using the instantaneous channel estimates
of the interfering users, see the \ac{UL} \ac{MU-MIMO} receiver structures used in, for example,
\cite{Hoydis:2013, Truong:13, Kong:2015, Abrardo:19}. Specifically, papers \cite{Hoydis:2013} and \cite{Abrardo:19}
proposed \ac{MMSE} receivers in block fading, whereas a \ac{MRC} and \ac{ZF} receiver in time-varying channels
in the presence of channel aging are used by \cite{Truong:13} and \cite{Kong:2015} respectively.
Note that the conceptual difference between the \ac{MRC} and \ac{ZF} receivers used in \cite{Truong:13}
and \cite{Kong:2015} and the \ac{MMSE} receiver proposed in \cite{Fodor:2021} lies in the fact that the
\ac{MMSE} receiver actively takes into account that the subsequent channel realizations are correlated
rather than adopting the \ac{MMSE} receiver structure developed for block fading channels.
Therefore, we refer to the \ac{MMSE} receiver in \cite{Fodor:2021} as an AR-aware receiver.

In the light of these works, it is natural to ask the following two questions:
\begin{itemize}
\item
What is the \ac{MU-MIMO} receiver that minimizes the \ac{MSE} of the received
data symbols in time-varying Rayleigh fading when all user channels are estimated and, therefore,
the multiuser interference does not need to be treated as noise?
\item
Can we calculate the average \ac{SINR} in the uplink of \ac{MU-MIMO} systems that employ the
above receiver, as a function of the number of \ac{MU-MIMO} users and receive antennas,
employed pilot and data powers and large scale fading?
\end{itemize}

Intuitively, finding the answers to these questions implies extending the results by (1) papers \cite{Hoydis:2013}
and \cite{Abrardo:19} (by generalizing some of those block fading results to \ac{AR} processes),
(2) papers \cite{Truong:13} and \cite{Kong:2015} (by developing the optimal linear receiver in \ac{MSE} sense)
and (3) paper \cite{Fodor:2021} (by not treating the \ac{MU-MIMO} interference as noise and deriving an \ac{SINR} formula
rather than using the \ac{MSE} as a performance metric).
Consequently, the objective of the present paper is to devise a \ac{MU-MIMO} receiver that
utilizes the channel estimates of each user and the fact that subsequent
channel coefficients are correlated in time. In other words, we propose and analyze a \ac{MU-MIMO}
receiver that is optimal in the presence of \ac{CSI} errors when the channel evolves in time
according to a Rayleigh fading autocorrelation process. It is also our objective to derive an
average \ac{SINR} formula that can serve as a basis for rate optimization schemes in future works.
Thus, our contributions 
to the existing literature summarized above and in Table \ref{tab:tab2} are two-fold:
\rev{
\begin{enumerate}
\item
Calculating the deterministic equivalent \ac{SINR} of the \ac{MU-MIMO} \ac{MMSE} receiver proposed in Proposition \ref{P2},
by proving Proposition 2, Theorem 2, whose proof is based on Theorem 1 and Corollary 1, is our main and novel result.
To the best
of our knowledge, Theorem 1, Lemma 4 (needed for Theorem 1) and Theorem 2 have not been published before.
\item
We would like to emphasize the usefulness of Proposition 3, which
gives a straightforward computation of the optimum pilot power
in a \ac{MU-MIMO} \ac{AR} Rayleigh fading environment as a root of a quartic equation.
\end{enumerate}
Our analytical (based on Theorem 2 and Proposition 3) and simulation results (comparing the performance of
the different \ac{MU-MIMO} receivers listed in Table IV)
indicate that the proposed AR-aware receiver
outperforms earlier \ac{AR} receivers in terms of the achieved \ac{SINR}, such as those
proposed by Truong and Heath \cite{Truong:13} and our own previously proposed scheme in \cite{Fodor:2021}.
}

The paper is organized as follows.
The next section describes our system model, which is similar to that used in, for example
\cite{Fodor:2021}, \cite{Hoydis:2013} or \cite{Truong:13}.
Section \ref{Sec:G} derives the MMSE receiver for autoregressive Rayleigh fading channels, stated
as Proposition \ref{P2}.
Section \ref{Sec:SINR} derives our key result, Theorem \ref{thm:1},
which can be considered as an extension of the \ac{SINR} results in \cite{Hoydis:2013} and \cite{Abrardo:19}
to \ac{AR} processes. The important feature of this implicit \ac{SINR} formula is that it does not
require to solve a system of equations or fixed point iterations due to the fact that the
implicit equation has a unique positive solution. Also, Subsection \ref{Sec:Opt} derives the
optimum pilot power in \ac{SU-MIMO} systems or in \ac{MU-MIMO} systems, in the special case when the large
scale fading components of all users are equal. The treatment of the optimum pilot
power in the general \ac{MU-MIMO} case is left for future work.
Section \ref{Sec:Num} discusses numerical results, and Section \ref{Sec:Conc} draws
conclusions.
\vspace{-2mm}
\section{System Model}
\label{Sec:Mod}

\subsection{Uplink Signal Model}
We consider a single cell \ac{MU-MIMO} system, where the \ac{BS} is equipped with
$N_r$ receive antennas, and there are $K$ uplink \acp{MS}.
(Note that typically $K \ll N_r$.) 
The \acp{MS} facilitate \ac{CSIR} acquisition at the \ac{BS} using orthogonal complex
sequences, such as the Zadoff-Chu sequences, defined as
$\mathbf{s} \triangleq \left[s_1,...,s_{\tau_p}\right]^T \in \mathds{C}^{{\tau_p \times 1}}$.
These pilot sequences satisfy 
$|s_i|^2 = 1$, for $i=1,..,\tau_p$ \cite{Sesia:11}.
To enable spatial multiplexing, the length of the pilot sequences
$\tau_p$ is chosen such that a maximum of $K$ users can be served simultaneously, implying that
$\tau_p \geq K$ holds.
In this \ac{MU-MIMO} system, 
$\tau_p$ subcarriers are used to construct the pilot sequences at each \ac{MS},
and $\tau_d$ subcarriers are used to transmit data symbols.
Each \ac{MS} has a total power budget
$P_{\text{tot}}$,
imposing the constraint
$\tau_p P_{p} + \tau_d P = P_{\text{tot}}$,
where $P$ is the transmit and $P_p$ denotes the pilot power. 
The trade-off between pilots and data signals as implied by the sum pilot and data power constraint 
has been studied
by several previous works, see for example \cite{LeviB, Ngo:14}.
In this paper, User-1 is the tagged user, while indexes $2 \ldots K$
are used to denote the interfering users from the tagged user's point of view.
Consequently, 
the received pilot signal transmitted by
User-1 at the \ac{BS} takes the form of \cite{Fodor:2021}:
\begin{align}
\mathbf{Y}^p(t)
&=
\alpha \sqrt{P_{p}}\mathbf{h}(t) \mathbf{s}^T +\mathbf{N}(t) ~~ \in \mathds{C}^{N_r \times \tau_p},
\label{eqn:received_training_seq}
\end{align}
\noindent where 
$\mathbf{h}(t) ~\in~\mathds{C}^{N_r \times 1} \sim \mathcal{CN}(\mathbf{0},\mathbf{C})$, that is,
$\mathbf{h}(t)$ is a 
complex normal distributed column vector
with mean vector $\mathbf{0}$ and covariance matrix $\mathbf{C}$. 
Furthermore, $\alpha$ denotes
large scale fading, and
$\mathbf{N}\in \mathds{C}^{N_r \times \tau_p}$
is the 
\ac{AWGN} with element-wise variance $\sigma_p^2$.

\subsection{Channel Model}
In this paper $\mx{h}$ denotes the complex channel which is modeled as a
stationary discrete time \ac{AR}(1) process as in \cite{Abeida:10, Hijazi:10, Fodor:2021}.
This model can be seen as a generalization of the block fading channel model:
$\mx{h}(t) = \mx{A} \mx{h}(t-1) + \boldsymbol{\vartheta}(t) \quad \in \mathds{C}^{N_r \times 1}$,
where $\boldsymbol{\vartheta}(t) \sim \mathcal{CN}\left(\mx{0},\bs{\Theta}\right)$
is the process noise vector
and $\mx{A}$ denotes the state transition matrix of the \ac{AR}(1) process \cite{Lehmann:08}.
In this paper we will use this \ac{AR}(1) model 
to approximate the Rayleigh fading channel.
We remark that the parameters of the \ac{AR}(1) model can be identified by existing
methods, such as those reported in \cite{McGuire:05, Krusevac:08} and 
\cite{Mekki:16}.
Due to the stationarity of $\mx{h}(t)$ 
we have
$\mx{C} = \mx{A} \mx{C} \mx{A}^H + \boldsymbol{\Theta}$. 

\subsection{Data Signal Model}
\begin{table}[t]
\caption{System Parameters}
\vspace{2mm}
\label{tab:notation}
\footnotesize
\begin{tabularx}{\columnwidth}{|X|X|}
\hline
\hline
\textbf{Notation} & \textbf{Meaning} \\
\hline
\hline
$K$ & Number of \ac{MU-MIMO} users \\
\hline
$N_r$ & Number of antennas at the BS \\
\hline
$\tau_p, \tau_d$ & Number of pilot/data symbols within a coherent set of subcarriers  \\
\hline
$\mx{s}\in \mathds{C}^{\tau_p \times 1}$ & Sequence of pilot symbols\\
\hline
$x$ & Data symbol \\
\hline
$P_p, P, P_{\text{tot}}$ & Pilot power per symbol, data power per symbol, and total power budget  \\
\hline
$\mx{Y}^p \in \mathds{C}^{N_r \times \tau_p}, y(t) \in \mathds{C}^{N_r}$ & Received pilot and data signal, respectively  \\
\hline
$\mx{h}(t), \hat{\mx{h}}(t) \in \mathds{C}^{N_r}$ & Fast fading channel and estimated channel \\
\hline
$\mx{A} \in \mathds{C}^{N_r \times N_r}$ & AR parameter of the channel\\
\hline
$\boldsymbol{\vartheta}(t) \in \mathds{C}^{N_r}, \bs{\Theta} \in \mathds{C}^{N_r \times N_r}$
& Process noise of the channel AR process and its covariance matrix\\
\hline
$\bs{\varepsilon}(t) \in \mathds{C}^{N_r}, \bs{\Sigma} \in \mathds{C}^{N_r \times N_r}$
& Channel estimation error and its covariance matrix\\
\hline
$\mx{G}, \mx{G}^\text{naive}, \mx{G}^\star$
& MU-MIMO receivers: generic, naive, and optimal, respectively. \\
\hline
\end{tabularx}
\end{table}
\vspace{-1mm}
Considering $K$ \ac{MU-MIMO} users,
the received data signal at the \ac{BS} at time 
$t$ is \cite{Fodor:2021}:
\begin{align}
\mathbf{y}(t)
&=
\underbrace{\mathbf{\alpha} \mathbf{h}(t) \sqrt{P} x(t)}_{\text{tagged user}}
+ \underbrace{\sum_{k=2}^K \mathbf{\alpha}_{k} \mathbf{h}_k(t) \sqrt{P_{k}} x_{k}(t)}_{\text{other users}}
+\mathbf{n}_d(t),
\label{eq:mumimo2}
\end{align}
\noindent where $\mathbf{y}(t)\in \mathds{C}^{N_r \times 1}$;
and
$\mathbf{\alpha}_{k} \mathbf{h}_k(t) \in \mathds{C}^{N_r \times 1}$
denotes the channel vector,
and $x_k(t)$ is the data symbol of User-$k$
transmitted at time $t$ with power $P_k$.
Furthermore $\mathbf{n}_d(t)~\sim \mathcal{CN}\left(\mx{0},\sigma_d^2\mx{I}_{N_r}\right)$
is the \ac{AWGN},
where $\mathbf{I}_{N_r}$ denotes the identity matrix of size $N_r$.

\subsection{Channel Estimation}
\label{Sec:Channel}
To acquire \ac{CSIR}, the \acp{MS} transmit orthogonal pilot sequences,
\color{black} and the \ac{BS}
uses \ac{MMSE} channel estimation based on~\eqref{eqn:received_training_seq}.
For algebraic convenience we define
\begin{align}
\mathbf{\tilde Y}^p(t)=\textbf{vec}\left(\mathbf{Y}^p(t)\right)=\alpha\sqrt{P_p} \mathbf{S} \mathbf{h}(t) +\mathbf{\tilde N}(t)\gf{,}
\end{align}
where $\textbf{vec}$ is the column stacking vector operator,
$\mathbf{\tilde Y}^p(t), \mathbf{\tilde N}(t) \in \mathds{C}^{\tau_p N_r \times 1}$
and
$\mathbf{S} \triangleq \mathbf{s}\otimes \mathbf{I}_{N_r} \in \tau_p N_r \times N_r$)
is such that $\mathbf{S}^H\mathbf{S}=\tau_p\mathbf{I}_{N_r}$.

\begin{lem}
\label{lem:mmsechannel}	
The MMSE channel estimator approximates the AR(1) channel based on the latest and the previous channel \qq{states} as
\begin{align}
\label{eq:hmmse}
\mathbf{\hat h}_{\textup{MMSE}}(t)
&=
\begin{bmatrix}
\mx{C} &
\mx{A} \mx{C}
\end{bmatrix}
\left( \frac{\sigma_p^2}{\alpha^2P_p \tau_p} \mx{I}_{2N_r} +  \mx{M}\right)^{-1} \nonumber \\
&~~~~
\left(\mathbf{\bar h}(t) + \frac{1}{\alpha\sqrt{P_p} \tau_p} \mathbf{\bar n}(t)\right),
\end{align}
where
$\mx{M} =\begin{bmatrix}
	\mx{C}  &
	\mx{A}\mx{C} \\
	\mx{C}\mx{A}^{H} &
	\mx{C}
	\end{bmatrix}$,
$\mx{\bar h}(t)=\begin{bmatrix}
\mx{h}(t) \\
\mx{h}(t-1)
\end{bmatrix}$ ~~and\vspace{3mm}\\
$\text{~~~}\mx{\bar n}(t)=\begin{bmatrix}
\mathbf{s}^H  \mx{N}(t) \\
\mathbf{s}^H \mx{N}(t-1)
\end{bmatrix}$.
\end{lem}

\rev{The proof is in Appendix A.}

\begin{cor}
\label{cor:rmmse}	
The estimated channel $\mathbf{\hat h}_{\textup{MMSE}}$ is a circular symmetric complex normal distributed vector
$\mathbf{\hat h}_{\textup{MMSE}}(t) \sim \mathcal{CN}(\mathbf{0},\mathbf{R}_{\textup{MMSE}})$,
with
\begin{align}
\label{eq:rmmse}
\mathbf{R}_{\textup{MMSE}} &= \mathds{E}_{\mathbf{h},\mathbf{n}} \{\mathbf{\hat h}_{\textup{MMSE}}(t) \mathbf{\hat h}_{\textup{MMSE}}^H(t)\} \nonumber \\
&=
\begin{bmatrix}
\mx{C} &
\mx{A} \mx{C}
\end{bmatrix}
\left( \frac{\sigma_p^2}{\alpha^2P_p \tau_p} \mx{I}_{2N_r} +  \mx{M}\right)^{-1}
\begin{bmatrix}
\mx{C} \\
\mx{C} \mx{A}^{H}
\end{bmatrix}
\\
&=
\left[
\begin{array}{ccc}
\mx{C} & \mx{AC}
\end{array}
\right]
\left[
\begin{array}{ccc}
\mx{C}+\mx{\Sigma} &  \mx{AC}\\
\mx{C}\mx{A}^H & \mx{C}+\mx{\Sigma}
\end{array}
\right]^{-1}
\left[
\begin{array}{ccc}
\mx{C}\\
\mx{C}\mx{A}^H
\end{array}
\right] ,
\nonumber
\end{align}
where $\mx{\Sigma}\triangleq \frac{\sigma_p^2}{\alpha^2P_p \tau_p} \mx{I}_{N_r}$.
\end{cor}

We note that \eqref{eq:rmmse} is obtained from \eqref{eq:hmmse} using
$\mathds{E}_{\mathbf{h},\mathbf{n}} \{\mx{\bar h}(t) \mx{\bar h}(t)^H\} = \mx{M}$ and
$\mathds{E}_{\mathbf{h},\mathbf{n}} \{\mx{\bar n}(t) \mx{\bar n}(t)^H\} = \tau_p \sigma_p^2 \mx{I}_{2N_r}$.
According to Corollary \ref{cor:rmmse}
and $\mathbf{h}(t) \sim \mathcal{CN}(\mathbf{0},\mathbf{C})$,
the covariance matrix of the channel estimation noise when using the \ac{MMSE} channel estimation is:
$\mx{Z} 
= \mx{C} - \mathbf{R}_{\textrm{MMSE}}$,
which is identical with the LS case discussed in \cite{Fodor:2021}, and we \qq{therefore}
omit the \ac{MMSE} subscript in the sequel.
\color{black}

\begin{lem}
\label{L2}
The channel realization $\mathbf{h}(t)$ conditioned on the
current and previous estimates
$\mathbf{\hat h}(t)$ and $\mathbf{\hat h}(t-1)$
is normally distributed as follows:
\begin{align}
\label{eq:ET}
\left(\mathbf{h}(t) \Big| \mathbf{\hat h}(t),\mx{\hat h}(t-1)\right)
&\sim
\mx{E} \bs{\zeta}(t)
+ \underbrace{\mathcal{CN}\Big(\mathbf{0},\mx{Z}\Big)}_{\textup{channel estimation noise}},
\end{align}
\noindent where for $\forall t$
\begin{align}
\label{eq:E}
&\bs{\zeta}(t) \triangleq
\left[
\begin{array}{ccc}
\mx{\hat h}(t)\\
\mx{\hat h}(t-1)
\end{array}
\right] \in \mathds{C}^{2N_r \times 1}, \nonumber \\
&\mx{E} \triangleq
\left[
\begin{array}{ccc}
\mx{C} & \mx{AC}
\end{array}
\right]
\left[
\begin{array}{ccc}
\mx{C}+\mx{\Sigma} &  \mx{AC}\\
\mx{C}\mx{A}^H & \mx{C}+\mx{\Sigma}
\end{array}
\right]^{-1}
\in \mathds{C}^{N_r \times 2N_r}, 
\end{align}
\begin{align}
\label{eq:T}
&\mx{Z} \triangleq \mx{C}-\mx{E}
\left[
\begin{array}{ccc}
\mx{C}\\
\mx{C}\mx{A}^H
\end{array}
\right] \in \mathds{C}^{N_r \times N_r},
\textup{~and~~} \nonumber \\
&\textup{Cov}\Big(\bs{\zeta}(t)\Big)=
\left[
\begin{array}{cc}
\mx{C+\Sigma} & \mx{AC} \\
\mx{CA^H} & \mx{C+\Sigma}
\end{array}
\right] \in \mathds{C}^{2N_r \times 2N_r}.~~~~~~
\end{align}
\end{lem}
\rev{The proof is in \cite{Fodor:2021}.}

\rev{\subsection{Summary}}
\rev{
This section described the system model consisting of a signal model and an \ac{MMSE} channel estimation scheme.
When the channel estimation is based on the current and previous channel observations
(i.e.\ $\mathbf{\hat h}(t)$ and $\mx{\hat h}(t-1)$), the conditional distribution of $\mathbf{h}$ is
complex normal with mean vector and covariance matrix according to Lemma \ref{L2}, which serves as a starting
point for deriving the optimal \ac{MU-MIMO} receiver in the sequel.}

\vspace{2mm}
\section{Deriving the MMSE Receiver for Time-Varying Rayleigh Fading Channels}
\label{Sec:G}
%
%
The \ac{BS} the transmitted data symbols by employing a
linear \ac{MMSE} receiver $\mathbf{G} \in \mathds C^{1 \times N_r}$,
which minimizes the \ac{MSE}
between the transmitted symbol $x$ and the estimated symbol $\mathbf{G} \mathbf{y}$:
\begin{align}
\label{eq:gstardef}
\mathbf{G}^\star
& \triangleq
\text{arg} \min_{\mathbf{G}} \mathds{E}_{\mx{h},\mx{n},x}\{ |\mathbf{G} \mathbf{y} - x|^2 \} ~~ \in \mathds C^{1 \times N_r}.
\end{align}
%

When the BS employs a naive receiver, it assumes perfect channel estimation,
and uses the estimated channel in place of the actual channel:
\begin{align}
\mathbf{G}^{\text{naive}} =
\alpha\sqrt{P}\mathbf{\hat h}^{H}(\alpha^2 P
\mathbf{\hat h}\mathbf{\hat h}^{H}+\sigma_d^2\mathbf{I}_{N_r})^{-1}.
\label{eqn:equalizer_definition_singlecell}
\end{align}
As we shall see, the naive receiver fails to minimize the MSE.

Next, we derive the \ac{MMSE} receiver vector $\mathbf{G}^\star$
that the receiver at the \ac{BS} should use to minimize the \ac{MSE} of the received data symbol $x$
of the tagged user based on the data signal $\mx{y}$.
Since the \ac{BS} can only use the estimated channels,
the objective function of this minimization must only depend on the
estimated channels $\mathbf{\hat h}(t)$ and $\mathbf{\hat h}(t-1)$.
This \ac{MMSE} receiver can be contrasted to the naive receiver,
which assumes that perfect \ac{CSIR} is available.
The \ac{MSE} of the received data symbols, as a function of the generic linear receiver $\mx{G}$ and the actual propagation channels $\mx{h}$,
was shown to have the following form \cite{FMT:15}:
%
\begin{align}
\label{eq:MSEGh}
&\text{MSE}\big(\mathbf{G}, \mathbf{H}\big)
=
\mathds{E}_{x,\mx{n}_d} \left\{|\mx{G}\mx{y}-x|^2\right\}  
=\left|\mathbf{G}  \alpha \mathbf{h} \sqrt{P}-1 \right|^2 \nonumber \\
&+{\sum_{k=2}^K P_{k}|\mathbf{G} \alpha_k \mathbf{h}_{k}|^2 }+ \sigma^2_d \mathbf{G} \mathbf{G}^H 
=1-\alpha \sqrt{P} \mathbf{G} \mathbf{h} - \alpha \sqrt{P} \mathbf{h}^H \mathbf{G}^H  \nonumber \\
&+\mathbf{G} \left(\sum_{k=1}^K \alpha^2_k P_{k}\mathbf{h}_{k} \mathbf{h}_{k}^H + \sigma^2_d \mathbf{I}_{N_r} \right) \mathbf{G}^H,
\end{align}
where
$\mathbf{H}=\left[\mathbf{h}_1, \dots, \mathbf{h}_K\right] \in \mathds{C}^{N_r \times K}$
collects the complex channel vector for each of the $K$ users.
%
We now seek to express the \ac{MSE} as a function of $\mx{G}$
and the estimated channel
$\hat{\mx{ H}}(t),\hat{\mx{ H}}(t-1)$, rather than the actual channel $\mx{H}$,
where the $\hat{\mx{H}}(t)$ and $\hat{\mx{ H}}(t-1)$ matrices collect the estimated channels.
To achieve this, we average the \ac{MSE} over
$\Big(\mathbf{h}_{k}|\hat{\mathbf{h}}_{k}(t),\hat{\mathbf{h}}_{k}(t-1)\Big)$ and obtain:
%
\begin{align}
&\text{MSE}\left(\mathbf{G}, \hat{\mathbf{H}}(t),\hat{\mathbf{H}}(t-1) \right)
=\mathds{E}_{\mathbf{H}|\hat{\mathbf{H}}(t),\hat{\mathbf{H}}(t-1)}\left\{\text{MSE}\left(\mathbf{G}, \mathbf{H} \right) \right\} \nonumber \\
&=
1- \alpha \sqrt{P} \mathbf{G} \mathbf{E} \bs{\zeta} - \alpha \sqrt{P}  \bs{\zeta}^H \mathbf{E}^H \mathbf{G}^H \nonumber \\
&+\mathbf{G}\left(\sum_{k=1}^K \alpha^2_k P_{k}
\left( \mathbf{E}_k  \bs{\zeta}_{k}  \bs{\zeta}^H \mathbf{E}_k^H \!+\! \mathbf{Z}_k\right)
\!+\! \sigma^2_d \mathbf{I}_{N_r} \right) \mathbf{G}^H,
\label{eq:msehath}
\end{align}
where the $\bs{\zeta}(t)$ vector and $\mx{E}$ and $\mx{Z}$ matrices, associated with the tagged user, were introduced in Lemma \ref{L2}, and $\bs{\zeta}_k(t)$, $\mx{E}_k$ and $\mx{Z}_k$ are the corresponding terms associated with user $k$.

%
We can now obtain
the following proposition: 
\begin{prop}
\label{P2}
The \textup{\ac{MU-MIMO}} \textup{\ac{MMSE}} receiver vector is given by:
\begin{align}
\label{eq:Gstar2}
\mx{G}^\star(t) &=
\textup{arg} \min_{\mx{G}} \textup{MSE}\left(\mx{G},\mx{\hat H}(t), \mx{\hat H}(t-1)\right) = \mx{b}^H(t) \mx{J}^{-1}(t),
\end{align}
where $\mx{b}(t)\in \mathds{C}^{N_r \times 1}$ and $\mx{J}(t) \in \mathds{C}^{N_r \times N_r}$
are defined as
\begin{align}
\label{eq:B3}
\mx{b}(t) &\triangleq \alpha \sqrt{P} 
\mx{E} \bs{\zeta}(t) 
,
\\
\label{eq:A3}
\mx{J}(t)
&\triangleq
\sum_{k=1}^K \alpha_k^2 P_k \left(\mx{E}_k \bs{\zeta}_k(t) \bs{\zeta}_k^H(t) \mx{E}_k^H + \mx{Z}_k\right) +   \sigma^2_d \mx{I}_{N_r}
.
\end{align}

\end{prop}
\rev{ Equation \eqref{eq:Gstar2} is a quadratic optimization problem and the proposition presents \qq{its solution}.
\qq{Specifically,} Proposition \ref{P2} states that the \ac{MU-MIMO} \ac{MMSE} receiver utilizes the estimated channels of all users
at both time $t$ and $t-1$, and the
\qq{$\mx{E}_k$ and $\mx{Z}_k$ matrices}
that were derived in Lemma \ref{L2}. \qq{To analyze the performance of this \ac{MU-MIMO} receiver,}
the next section uses the results of this section as a starting point, and will calculate the average \ac{SINR}, as the main result of this paper,
using random matrix theory.}

\section{Calculating the 
\ac{SINR} of the Received Data Symbols} 
\label{Sec:SINR}
\subsection{Determining the Instantaneous SINR with $\mx{G}^\star$}
Based on the received signal $\mx{y}$, the \ac{BS} employs the linear receiver $\mx{G}$ to estimate the transmitted symbol of the tagged user as:
$\hat{x}=\mathbf{G}\mathbf{y}$.
The expected energy of $\hat{x}$, conditioned on $\big(\hat{\mathbf{H}}(t),\hat{\mathbf{H}}(t-1)\big)$,
is expressed as:
\begin{equation} \nonumber
\begin{aligned}
\label{eq:estsymbol}
&
\mathds{E}_{x,\textbf{n}_d,\mathbf{H}|\hat{\mathbf{H}}(t),\hat{\mathbf{H}}(t-1)}
\!\left\{\left|\hat{x}\right|^2\right\}
=
\alpha^2 P |\mathbf{G} \mathbf{E} \bs{\zeta}(t)|^2 \nonumber \\
&
+\!\sum_{k=2}^K \alpha_k^2 P_k |\mathbf{G} \mathbf{E}_k \bs{\zeta}_k(t)|^2 
+\!\underbrace{\sum_{k=1}^K \alpha_k^2 P_k  \mathbf{G} \mathbf{Z}_k \mathbf{G}^H}_{\mbox{\footnotesize ch. estim. noise}}
\!+\!\sigma_d^2 \mathbf{G}  \mathbf{G}^H.
\end{aligned}
\end{equation}
We can now state the following lemma, which determines the instantaneous \ac{SINR}.
\vspace{-1mm}
\begin{lem}
\label{lem:1}
Assume that the receiver employs \textup{\ac{MMSE}} symbol estimation.
Then the instantaneous \ac{SINR} of the estimated data symbols,
$\gamma\Big(\mathbf{G}^\star, \hat{\mathbf{H}}(t),\hat{\mathbf{H}}(t-1)\Big)$
is given as:
\begin{equation}
\label{eq:lemma2Eq}
\gamma\Big(\mathbf{G}^\star(t), \hat{\mathbf{H}}(t),\hat{\mathbf{H}}(t-1)\Big)
=
\alpha^2 P \bs{\zeta}^H(t) \mx{E}^H \mathbf{J}_1^{-1}(t)  \mx{E} \bs{\zeta}(t),
\end{equation}
where
$\mathbf{J}_1(t) \triangleq \mathbf{J}(t)-\alpha^2 P \mx{E} \bs{\zeta}(t) \bs{\zeta}^H(t) \mx{E}^H$.
\end{lem}
\vspace{-1mm}
\noindent The lemma is obtained when $\mathbf{G}^\star(t)$ (c.f. \eqref{eq:Gstar2}) is substituted into \eqref{eq:lemma2Eq}.
\ignore{>>>>>>>>>>>>>>>>>>>>>>>>>>>>>>>>>>>>>>>>>>>>>>>>>>>>>>>
\begin{proof}
The proof is given in Appendix \ref{Sec:AppIV}.
\end{proof}
The subsequent subsections are concerned with calculating the average \ac{SINR} when averaging $\gamma$ in \eqref{eq:lemma2Eq} over the channel realizations and transmitted symbols over all users.
<<<<<<<<<<<<<<<<<<<<<<<<<<<<<<<<<<<<<<<<<<<<<<<<<<<<<<<<<}

\subsection{Calculating the Average \ac{SINR}}
To calculate the average \ac{SINR}, we first make the following considerations.
According to \eqref{eq:B3},
$\mx{b}_k(t) = \alpha_k \sqrt{P_k} \mx{E}_k \bs{\zeta}_k(t)$.
that is
$\mx{b}_k \sim \mathcal{CN}(0,\bs{\Phi}_k)$,
where, $\bs{\Phi}_k$ can be calculated using the covariance matrix $\bs{\zeta}$ in \eqref{eq:T} as: 
\begin{align}
\label{eq:phidef}
\bs{\Phi}_k &= \alpha_k^2 P_k \mx{E}_k
\left[
\begin{array}{ccc}
\mx{C}_k+\bs{\Sigma}_k & \mx{A}_k \mx{C}_k \\
\mx{C}_k \mx{A}_k^H & \mx{C}_k +\bs{\Sigma}_k
\end{array}
\right]
\mx{E}_k^H \nonumber \\
&= \alpha_k^2 P_k \mx{E_k}
\left[
\begin{array}{c}
\mx{C}_k \\ \mx{C}_k\mx{A}_k^H
\end{array}
\right]
.
\end{align}
Notice that:
\vspace{-2mm}
\begin{align}
\mathbf{J}_1(t) &= \mathbf{J}(t)-\alpha^2 P \mx{E} \bs{\zeta}(t) \bs{\zeta}^H(t) \mx{E}^H 
= \underbrace{\sum_{k=2}^K \mx{b}_k \mx{b}_k^H}_{\triangleq \mx{B}\mx{B}^H} \nonumber
\end{align}
\vspace{-4mm}
\begin{align}
\label{eq:betadef}
&+ \underbrace{\sum_{k=1}^K \alpha_k^2 P_k \mx{Z}_k + \sigma_d^2 \mx{I}_{N_r}}_{\triangleq\boldsymbol{\beta}},
\end{align}
where
$\boldsymbol{\beta} \in \mathds C^{N_r \times N_r}$
is a constant matrix (with measurable elements) and the
$\mx{b}_k$
vectors are characterized by the $\mx{\hat{h}}_k(t)$, $\mx{\hat{h}}_k(t\!-\!1)$ estimated channels.
Substituting $\mx{b}_k$ in \eqref{eq:lemma2Eq} yields
\begin{align}
\label{eq:gamma}
\gamma\Big(\mathbf{G}^\star(t), \mx{\hat{H}}(t), \mx{\hat{H}}(t-1)\Big)
&=
\mx{b}^H \left(\mx{B}\mx{B}^H + \boldsymbol{\beta} \right)^{-1} \mx{b},
\end{align}
where we recall that we drop the index of the tagged user (User-1), that is
$\mx{b} \triangleq \mx{b}_1$.
\rev{For block fading channels, reference \cite{Hoydis:2013} suggests that the deterministic equivalent of the \ac{SINR}
is a good approximation of the average \ac{SINR} in the \ac{MU-MIMO} system when the number of antennas is greater than
a certain number. This result motivates us to determine the deterministic equivalent \ac{SINR} also for our system,
in which the channels evolve according to an \ac{AR} process. As we shall see, the deterministic equivalent is a good
approximation of the average \ac{SINR} also in our case. To this end,
we can now state the following proposition, which calculates the deterministic equivalent \ac{SINR} for \ac{AR} channels.}
\vspace{-1mm}
\begin{prop}
\label{prop:Hoydis}
Assume that 
\begin{align*}
N_r\to\infty \text{~~and~~} \limsup_{N_r\to\infty} K/N_r&<\infty,
\end{align*}
then, for the instantaneous \ac{SINR} of the tagged user, denoted as $\gamma$, the following holds:
\begin{align}
\gamma - \textup{tr}\Big( \mx{\Phi} \mx{T}\Big)
~~\xrightarrow[N_r\rightarrow\infty]{\text{a.s.}} 0,
\label{eq:HoydisTh1}
\end{align}
\noindent where $\mx{T}$ is defined as
\begin{align}
\label{eq:Tdef}
\mx{T} & \triangleq \left(\frac{1}{N_r} \sum_{k=2}^K
\frac{ \mx{\Phi}_k }{1+\delta_{k}}
+ \boldsymbol{\beta}\right)^{-1},
\end{align}
and $\delta_{k}$, for $k=2,\ldots,K$ are the solution of the equation system defined by:
\begin{align}
\label{eq:deltak}
\delta_{k} &=
\frac{1}{N_r}
\textup{tr}\left( \mx{\Phi}_k \left(\frac{1}{N_r} \sum_{\ell=2}^K \frac{\mx{\Phi}_\ell}{1+\delta_{\ell}}
+\boldsymbol{\beta} \right)^{-1}\right).
\end{align}
\end{prop}
\begin{proof}
The proof is in Appendix \ref{App:Hoydis}.
\end{proof}
Note that
According to \cite{Hoydis:2013},
$\delta_{k}$ ($k=2,\ldots,K$) can be obtained by fixed point iteration starting from $\delta_{k}=1/\sigma_d^2$ ($k=2,\ldots,K$).
Based on the above proposition, for finite $N$, we can write that:
\begin{align}
\label{eq:gammaT}
\bar \gamma  \approx \text{tr}\Big( \mx{\Phi} \mx{T}\Big).
\end{align}

It is worth noting that determining the average SINR for a single user requires to
solve the above system of equations,
because calculating
$\delta_k$ for $k=1$ is inter\-twined with calculating the $\delta_k$:s for $k=2 \dots K$
in \eqref{eq:deltak}.
This observation motivates us to seek 
an alternative solution,
according to which calculating the \ac{SINR} for the tagged user does not require to solve
a system of equations.
We note that a more restricted special case assuming identical user settings
for the block fading model was studied in \cite{Hoydis:2013}.
\rev{Regarding the complexity of \qq{determining the \ac{SINR}} and the number of iterations needed, we make the
following observation.}

\begin{observ}
\rev{
    The complexity of one iteration of the fixed point iteration algorithm used to solve the system of $K-1$ equations \eqref{eq:deltak} is
    $\mathcal{O}(KN_r^{2.37})$ and the number of iterations needed in order to get an estimate of the \ac{SINR} with error less than
    or equal to some $\epsilon$ is $\mathcal{O}(log(1/\epsilon))$. In conclusion, the time complexity of the fixed point iteration algorithm used to find
    the \ac{SINR} of one user is $\mathcal{O}(KN_r^{2.37}\log(1/\epsilon))$.
}
\end{observ}
\begin{proof}
\color{black}
It is shown in \cite{Wagner:2012}, that the system of equations in Proposition 2 has
a unique positive solution and the fixed point iteration converges to this solution when it is started from the initial point $\delta_k=1/\sigma_d^2 (k=2,\ldots,K)$.
Regarding the complexity of the iteration, notice that \qq{on the} right hand side of \eqref{eq:deltak}
the term that is inverted is the same for every value of $k$, and \qq{needs to be computed once} during every iteration step.
\qq{To compute this term,} we need to add $\mathcal{O}(K)$ number of $N_r \times N_r$ matrices,
and hence the complexity is $\mathcal{O}(KN_r^2)$.
Next, to invert this term, we use the well-known Coppersmith-Winograd algorithm of complexity $\mathcal{O}(N_r^{2.37})$.
We can now calculate the matrix product inside the trace operation for every $K$;
once again using the Coppersmith-Winograd algorithm, this step has complexity $\mathcal{O}(KN_r^{2.37})$.
\qq{Finally,} computing the trace for \qq{each} $k$ has complexity $\mathcal{O}(KN_r)$.
In conclusion, the complexity of one iteration step is $\mathcal{O}(KN_r^2 + N_r^{2.37} + KN_r^{2.37} + KN_r) = \mathcal{O}(KN_r^{2.37})$.
Regarding the number of iterations needed, by equation (111) in \cite{Wagner:2012}, the $\delta_k$ \qq{converges} exponentially to the fixed point.
\qq{Consequently, the number of iterations needed to reach precision $\epsilon$ is $\mathcal{O}(\log(1/\epsilon))$.}
In conclusion, calculating the \ac{SINR} of a single user in a system with $K$ users and $N$ antennas,
to a precision of $\epsilon$, is $\mathcal{O}(KN_r^{2.37}\log(1/\epsilon))$.
\color{black}
\end{proof}
\rev{By the numerical experiments reported in Section V, we found that the procedure converges
in less than 10 iterations in all investigated scenarios.}

\subsection{Calculating the Average \ac{SINR} in the Case of Independent and Identically Distributed Channel Coefficients}
\label{Sunsec:Uncorr}
If the $N_{r}$ antennas are sufficiently spaced apart,
the correlation matrix $\mathbf{C}_k$ of the channel of User-$k$ can be assumed
to be of the form of $\mathbf{C}_k=c_k \mathbf{I}_{N_{r}}$.
Additionally using $\mathbf{\Sigma}_k=s_k \mathbf{I}_{N_{r}}= \frac{\sigma_p^2}{\alpha_k^2P_{p,k} \tau_{p,k}} \mx{I}_{N_r}$, based on the definition of $\mx{E}_k$ in \eqref{eq:E} we have:
\begin{align}
\label{eq:Eidentity}
\mx{E}_k
&=
\left[
\begin{array}{ccc}
\hat{e}_k\mx{I}_{N_r} & \check{e}_k\mx{I}_{N_r}
\end{array}
\right]~\in~\mathds{C}^{N_r \times 2N_r},
\end{align}
where:
\begin{align}
\label{eq:es}
\hat{e}_k &= \frac{c_k(c_k+s_k-a_kc_ka_k^*)}{c_k(c_k+s_k-a_kc_ka_k^*)+s_k(c_k+s_k)}, 
\text{~~and~~~~} \nonumber \\
\check{e}_k &= \frac{a_kc_ks_k}{(c_k+s_k)^2-a_kc_k^2a_k^*}.
\end{align}
Furthermore, due to the definition of $\mx{Z}_k$ in \eqref{eq:T}, we have that
$\mx{Z}_k = z_k \mx{I}$, where
\begin{align}
z_k &= \frac{c_ks_k(c_k+s_k-a_kc_ka_k^*)}{(c_k+s_k)^2-a_kc_k^2a_k^*}.
\end{align}
Additionally,
\begin{align}
\label{eq:Phi}
\bs{\Phi}_k &= \phi_k \mx{I}_{N_r},
\text{~~with~~}
\phi_k =
\alpha_k^2 P_k (\hat{e}_k c_k +  \check{e}_k c_k a_k^*).
\end{align}
From \eqref{eq:Eidentity} and the definition of $\mx{b}_k(t)$ in \eqref{eq:B3}, we get:
\begin{align}
\label{eq:BI}
\mx{b}_k(t)
&=
\alpha_k \sqrt{P}_k \left(\hat{e}_k \mx{\hat h}_k(t) + \check{e}_k \mx{\hat h}_k(t-1) \right)
 \quad \in \mathds{C}^{N_r \times 1}.
\end{align}
Using these definitions, the constant matrix $\bs{\beta}$ in the
\ac{SINR} expression of the tagged user (in \eqref{eq:gamma}) becomes:
$\bs{\beta} = \beta \mx{I}_{N_r}, \text{~where}: 
~~\beta \triangleq \sum_{k=1}^K \alpha_k^2 P_k z_k + \sigma^2_d$.
The average \ac{SINR} for the tagged user $(k=1)$ is then calculated as:
\begin{equation}
\label{eq:averageSINR}
\bar{\gamma}
= \mathds{E}_{\mathbf{b}_k, k=1\ldots K} \left\{\mx{b}^H \left(\sum_{k=2}^K \mathbf{b}_k \mathbf{b}_k^H
+ \beta \mathbf{I}_{N_{r}}\right)^{-1} \mx{b}\right\},
\end{equation}

To calculate the average \ac{SINR}, \rev{notice
that random matrices of the form $\mx{v}\mx{v}^H$ (\qq{a.k.a.} random dyads) with
$\mx{v} \sim \mathcal{CN}\left(\mx{0},\lambda \mx{I}_n\right)$ (where $n$ is large)
play a central role in \eqref{eq:averageSINR}.
It has been shown in several important works in the field of random matrices,
that the asymptotic distribution of the eigenvalues
can be advantageously used to deal with such matrices \cite{Couillet:11, Couillet:12, Muller:13}.
In particular, the Stieltjes transform
is often used to characterize the asymptotic distribution of the eigenvalues
of large dimensional random matrices \cite{Couillet:11, Wagner:2012, Muller:13}.
As it is discussed in details in \cite{Couillet:11, Couillet:12, Wen:13, Zhang:13},
from a wireless communications standpoint, the Stieltjes transform
can be used to characterize the
\ac{SINR} of multiple antenna communication models, including the \ac{MU-MIMO} interference broadcast
\qq{and multiple access channels}.
The Stieltjes transform of random variable $X$ with \ac{CDF} $P_X(x)$ is defined as
\begin{align}
\label{eq:defg}
G_X(s)\triangleq \mathds{E}\left\{\frac{1}{X-s}\right\} = \int_x  \frac{1}{x-s} d P_X(x).
\end{align}
}
\rev{The $\mathcal{R}$-transform is closely related to the Stieltjes transform by the following relation
\begin{align}
\label{eq:defr}
\mathcal{R}_X(s)\triangleq G_X^{-1}(-s) - \frac{1}{s},
\end{align}
where $G^{-1}(-s)$ denotes the inverse function of the Stieltjes transform \cite{Muller:13}.
The $\mathcal{R}$-transforms
are commonly used to provide approximations of capacity expressions in large dimensional systems, see e.g. \cite{Tulino:05, Muller:13}.
In the present work, the relationship between the Stieltjes and $\mathcal{R}$-transforms
will be used to provide a deterministic approximation of the average \ac{SINR} in \ref{eq:averageSINR}.}
\rev{The main reason for using the $\mathcal{R}$-transform is its additive property,
\qq{according to which} $\mathcal{R}_{X+Y}(s) = \mathcal{R}_X(s) + \mathcal{R}_Y(s)$.}
\rev{To calculate the deterministic approximation, we first prove an important \qq{theorem}, which, together with its corollary concerning
the $\mathcal{R}$-transform of random dyads of the type $\mx{v}\mx{v}^H$ will be important
\qq{in} calculating the average \ac{SINR} in the sequel.}

\begin{thm}
\label{thm:2}
Let $\lambda_i$ be a bounded sequence $\lambda_i < \lambda_{\max}$ such that
\begin{align}
\lim_{n\rightarrow \infty} \frac{\lambda_1 + \lambda_2 + \ldots + \lambda_n}{n} &= \bar{\lambda}.
\end{align}
Furthermore, let $\mx{v}^{(n)}$ be a sequence of complex normal distributed random vectors with $\mx{0}$ means and covariances
$\mx{R}_n = diag(\lambda_1, \lambda_2, \ldots \lambda_n)$.
Denote by $\omega_n$ a randomly selected eigenvalue of the dyad $\mx{v}^{(n)}\left(\mx{v}^{(n)}\right)^H$.
Then \rev{the limit of the $\mathcal{R}$-transform of the distribution of $\omega_n$ is given as follows:}
\vspace{-2mm}
\begin{align}
\label{eq:Th1}
\rev{\lim_{n\rightarrow \infty} \mathcal{R}_{\omega_n}\left( \frac{s}{n} \right)} &= \frac{\bar{\lambda}}{1 - s\bar{\lambda}}.
\end{align}
\end{thm}
\begin{proof}
The proof is in Appendix \ref{Sec:AppVI}.
\end{proof}
From Theorem \ref{thm:2}, the following result is immediate:
\begin{cor}
\label{cor:rtrafo}
Let the vector $\mx{v} \sim \mathcal{CN}\left(\mx{0},\lambda \mx{I}_n\right)$.
\rev{The $\mathcal{R}$-transform of the distribution of a randomly selected eigenvalue
of $\mx{v}\mx{v}^H$, denoted by $\omega_n$}
\rev{ is asymptotically equal to:}
\begin{align}
\lim_{n\rightarrow \infty} R_{\omega_n}\left(\frac{s}{n}\right) = \frac{\lambda}{1 - s\lambda}.
\end{align}
\end{cor}
For finite $n$, Corollary \ref{cor:rtrafo} gives the approximation $R_{\omega_n}\left(s\right) \approx \frac{\lambda}{1 - ns\lambda}$,
which we will use in our proof of Theorem \ref{thm:1}.
The following theorem, which is our main result,
states the average \ac{SINR} in the presence of a per user total power budget.
\begin{thm}
\label{thm:1}
The asymptotic average \ac{SINR} $\bar{\gamma}$, \rev{that is $\bar{\gamma}$ as $N_r \rightarrow \infty$,}
is the unique positive solution to the following equation:
\begin{align}
\underbrace{\sum_{k=1}^K \alpha_k^2 P_k z_k + \sigma^2_d}_{\beta}  &=
\frac{N_r \phi}{\bar{\gamma}}-
\sum_{k=2}^K \frac{\phi_k}{1+\frac{\bar{\gamma} \phi_k}{\phi}}.
\label{eq:SINR35}
\end{align}
\end{thm}
\begin{proof}
The proof is in Appendices \ref{Sec:AppV} 
and \ref{Sec:AppVII}.
Specifically, we provide two alternative proofs to Theorem \ref{thm:1}, both of which rely on random matrix
considerations, and have their own merits.
The first proof invokes the Stieltjes and $\mathcal{R}$-transforms of
probability distributions (Appendix \ref{Sec:AppV}),
while the second proof (Appendix \ref{Sec:AppVII})
uses the results in \cite{Wagner:2012} and relies on a matrix trace approximation as in
the lemmas invoked by both \cite{Truong:13} and \cite{Hoydis:2013}.
\end{proof}
\vspace{-2mm}

Notice that the $\phi_k$:s in Theorem \ref{thm:1}
can be easily calculated by means of \eqref{eq:Phi}, as long as
the covariances matrices of the channels ($\mx{C}_k$) and the transition matrices of the autoregressive
process that characterize the channels ($\mx{A}_k$) are accurately estimated. Therefore, the average
\ac{SINR} of the tagged user can be calculated by solving \eqref{eq:SINR35}, rather than solving a system
of equations as in Proposition \ref{prop:Hoydis}.
In the numerical section, we will
investigate the impact of AR parameter estimation errors on the average \ac{SINR} performance.
\subsection{Optimum Pilot Power}
\label{Sec:Opt}
In this subsection, we determine the
optimum pilot power in \ac{SU-MIMO} systems and in \ac{MU-MIMO} systems in the special case
when the large scale fading components of all users are equal.
\rev{ By deriving a closed form expression for the optimum pilot power,
we learn that it does not depend on the number of antennas $N_r$.}
The treatment of the optimum pilot power in the general case, in which the large scale fading
components are different is left for future work.

In the case in which each user has the same path loss $\alpha_k = \alpha~\forall k$,
channel covariance matrix $\mx{C}_k = \mx{C}=c\mx{I}~\forall k$, and \ac{AR} parameter $a_k=a~\forall k$,
equation \eqref{eq:SINR35} of Theorem \ref{thm:1} simplifies to
\begin{align}
\label{eq:gamma_special}
\frac{\beta}{\phi}
&=
\frac{N_r}{\bar{\gamma}}-\frac{K-1}{1+\bar{\gamma}}.
\end{align}
It follows from Theorem \ref{thm:1} that finding the optimum pilot power, which maximizes the average \ac{SINR}
in the \ac{SU-MIMO} case, that is when $K=1$, is equivalent with maximizing 
$\frac{\phi}{\beta}$.
In the \ac{MU-MIMO} case ($K>1$), we can first state the following interesting result.
\begin{lem}
\label{Lem4}
Assume $K>1$ and that each user employs the same pilot-to-data power ratio,
and, consequently, achieves the same \ac{SINR}.
The optimum pilot and data powers are given as the solution of the following maximization problem:

\begin{equation}
\label{eq:phiperbeta}
\begin{aligned}
& \underset{P,P_{p}}{\textup{maximize}}
& & \frac{\phi}{\beta} 
~~~~~~\textup{subject to}
~~P \tau_d + P_{p} \tau_p = P_{\textup{tot}}.
\end{aligned}
\end{equation}

\end{lem}
\begin{proof}
The right hand side of \eqref{eq:gamma_special} is strictly decreasing in $\bar{\gamma}$
since
\begin{align}
\frac{\partial}{\partial \bar{\gamma}}\left( \frac{N_r}{\bar{\gamma}}
-\frac{K-1}{1+\bar{\gamma}} \right)
&= -\frac{N_r}{\bar{\gamma}^2} + \frac{K-1}{(1+\bar{\gamma})^2} ~~~~ \nonumber \\
&< \frac{-N_r + K - 1}{\bar{\gamma}^2} < 0.
\end{align}
Hence, $\bar{\gamma}$ is strictly decreasing in the left hand side of \eqref{eq:gamma_special}
with respect to
$\frac{\beta}{\phi}$,
from which the lemma follows.
\end{proof}
\rev{To get some intuition behind this Lemma,
recall from equation \eqref{eq:phidef} that $\phi$ is the expected power of the estimated received data symbol.
Furthermore,
$\beta = \sum_{k=1}^K \alpha_k^2 P_k z_k + \sigma^2_d$, that is the sum of the data powers times the channel estimation errors
and the power of the data symbol noise.
Hence, the ratio $\phi / \beta$ \qq{reflects}
the ratio of the powers of the useful and the non-useful information arriving at the receiver.}

A consequence of this lemma is that the optimal pilot power is invariant under the number of antennas $N_r$,
since $N_r$ does not appear in the optimization problem \ref{eq:phiperbeta}. This observation will be confirmed
in the numerical section (see Figure \ref{Fig:Fig4}).

We now state the following proposition, which will provide some useful insights
in the impact of optimum pilot power setting in the numerical section.
\begin{prop}
\label{prop:OptP3}
In a \ac{MU-MIMO} system, in which each user has the same path loss,
and $a \in \mathds{R}$,
the optimal pilot power is a positive real root in the interval $\left(0,\frac{P_{\textup{tot}}}{\tau_p}\right)$
of the following quartic equation:
\begin{align}
\label{eq:OptP3}
c_0 + c_1 P_p + c_2 P_p^2 + c_3 P_p^3 + c_4 P_p^4 &= 0,
\end{align}
where
{\small
\begin{align*}
    c_4 &= (a^2 - 1)^2 c^3 \alpha^6 (K\sigma_p^2 - \sigma_d^2 \tau_d) \tau_p^4;~~ \\
    c_3 &= 2 (a^2 - 1) c^2 \alpha^4 \sigma^2_p ((a^2 - 1) c K P_{\textup{tot}} \alpha^2 - K \sigma_2^p + 2 \sigma_2^d \tau_d) \tau_p^3; \\
    c_2 &= c \alpha^2 \sigma_2^p ((a^2 - 1)^2 c^2 K P_{\textup{tot}}^2 \alpha^4 + \sigma^2_p ((1 + a^2) K \sigma_2^p + 
           (a^2 - 5) \sigma^2_d \tau_d) \nonumber \\
        &~~~+ (a^2 - 1) c P_{\textup{tot}} \alpha^2 (4 K \sigma^2_p + (a^2 - 1) \sigma^2_d \tau_d)\tau_p^2; \\
    c_1 &= -2 \sigma_p^4 ((a^2 - 1) c P_{\textup{tot}} \alpha^2 + \sigma^2_p + a^2 \sigma^2_p)\cdot 
    (c K P_{\textup{tot}} \alpha^2 + \sigma^2_d \tau_d) \tau_p; \\
    c_0 &= (a^2 + 1) P_{\textup{tot}} \sigma_p^6 (c K P_{\textup{tot}} \alpha^2 + \sigma^2_d \tau_d).
\end{align*}
}
\end{prop}
\qq{The proof is in Appendix F.}
\rev{\subsection{Summary}}
\rev{This section developed a method to calculate the average \ac{SINR} in \ac{MU-MIMO} systems that use the receiver
proposed in Proposition \ref{P2}.
For the general case, when the antenna coefficients are correlated, Proposition \ref{prop:Hoydis}
gives the deterministic equivalent of the \ac{SINR} and, according to \eqref{eq:gammaT}, it gives a good approximation of the
average \ac{SINR} when the number of antennas is large.
For the special case, when the channel coefficients are independent and identically distributed, Theorem \ref{thm:1}
gives the average \ac{SINR} and, by further assuming the special case of all users
having the same large scale fading, the optimum pilot power is given by Proposition \ref{prop:OptP3}.
These results will be verified by simulations and illustrated by numerical examples in the next section.}

\section{Numerical Results}
\label{Sec:Num}
\begin{table*}[ht]
\caption{System Parameters}
\label{tab:params}
\footnotesize
\centering
\begin{tabular}{|l|l|}
\hline
\hline
\textbf{Parameter}                     & \textbf{Value} \\
\hline
\hline
\ac{AR} state transition matrix $\mx{A}=a\mx{I}_{N_r}$   & $a=0, 0.1, \dots 0.95$ \\ \hline
Number of receive antennas at the \ac{BS}    & $N_r=20, 100$  \\ \hline
Path loss of tagged \ac{MS}              & $\alpha=90$ dB \\ \hline
Number of data and pilot symbols       & $\tau_d=11;{~}\tau_p=1$ \\ \hline
Sum pilot and data power constraint     & $\tau_p P_p+\tau_d P=P_{\text{tot}}$ =250 mW. \\ \hline
MIMO receivers                         & $\text{Naive, MRC, conventional,AR-aware with covariances,AR-aware proposed MMSE}$ \\ \hline
Number of users                         & $K=1,3,10,20,50$ \\ \hline
\hline
\end{tabular}
\end{table*}

\begin{table*}[ht]
\vspace{2mm}
\caption{MU-MIMO Receivers}
\label{tab:G}
\footnotesize
\centering
\begin{tabular}{|l|p{0.6\textwidth}|}
\hline
\hline
\textbf{Receiver}                     & \textbf{Description} \\
\hline
\hline
Naive receiver: $\mx{G}^{\text{naive}}$ & Assumes perfect channel estimation and block fading \cite{Eraslan:2013}. \\ \hline
Conventional with covariances           & Uses the cov. matrix of interfering users, treats interference as noise and assumes block fading\cite{FMT:15}. \\ \hline
AR-aware with covariances               & Uses Kalman assisted channel est. for the tagged user, treats interference as noise,
                                            uses an \ac{AR} channel model \cite{Fodor:2021}. \\ \hline
Conventional with inst. ch. est. (Hoydis)      & Uses channel estimates for all users and assumes block fading \cite{Hoydis:2013, Li:15, Abrardo:19}. \\ \hline
Maximum ratio combining (Troung-Heath)  & MRC receiver with/out Kalman filtering and channel prediction, uses \ac{AR} channel models \cite{Truong:13}. \\ \hline
Proposed in the present paper:  $\mx{G}^{\star}$
                                        & Uses Kalman filter assisted channel est. for all users, uses an \ac{AR} channel model. \\ \hline
\hline
\end{tabular}
\end{table*}

To obtain numerical results, we study a single cell \ac{MU-MIMO} system, in which the \acp{MS} are equipped with a
single transmit antenna, while the \ac{BS} is equipped with $N_r$ receive antennas.

We study the case in which the channel coefficients are of the complex channel vector are independent and identically distributed
as described in Subsection \ref{Sunsec:Uncorr}.
The most important parameters of this system that must be properly set to generate numerical results using the \ac{SINR}
derivation in this paper (utilizing Proposition \ref{prop:Hoydis} and Theorem \ref{thm:1})
are listed in Table \ref{tab:params}.
To benchmark the performance of the proposed \ac{MU-MIMO} receiver, we use the conventional \ac{MMSE} receivers, see table \ref{tab:G}.
An \ac{AR}-aware receiver was proposed in our previous work \cite{Fodor:2021}, in which the receiver does not
utilize the instantaneous channel estimates of the interfering users, but treats interference as noise through
the channel covariance matrices.
In order to demonstrate
the gain due to using the channel
estimate of each user, we compare the SINR performance of the proposed \ac{MU-MIMO} receiver in this paper
with that developed in \cite{Fodor:2021}. We also use the \ac{MRC} receiver that was used in the context
of channel aging by \cite{Truong:13}. The \ac{MRC} receiver in \cite{Truong:13} was used (1)
with MMSE channel estimation based on the current observation only, (2) with Kalman filter forecast and
(3) channel prediction using a $p$-order Kalman filter. For benchmarking purposes, we will consider all three
variants of the scheme used by Troung and Heath in \cite{Truong:13}.
\begin{figure}[ht]
\centering
\includegraphics[width=1.\columnwidth]{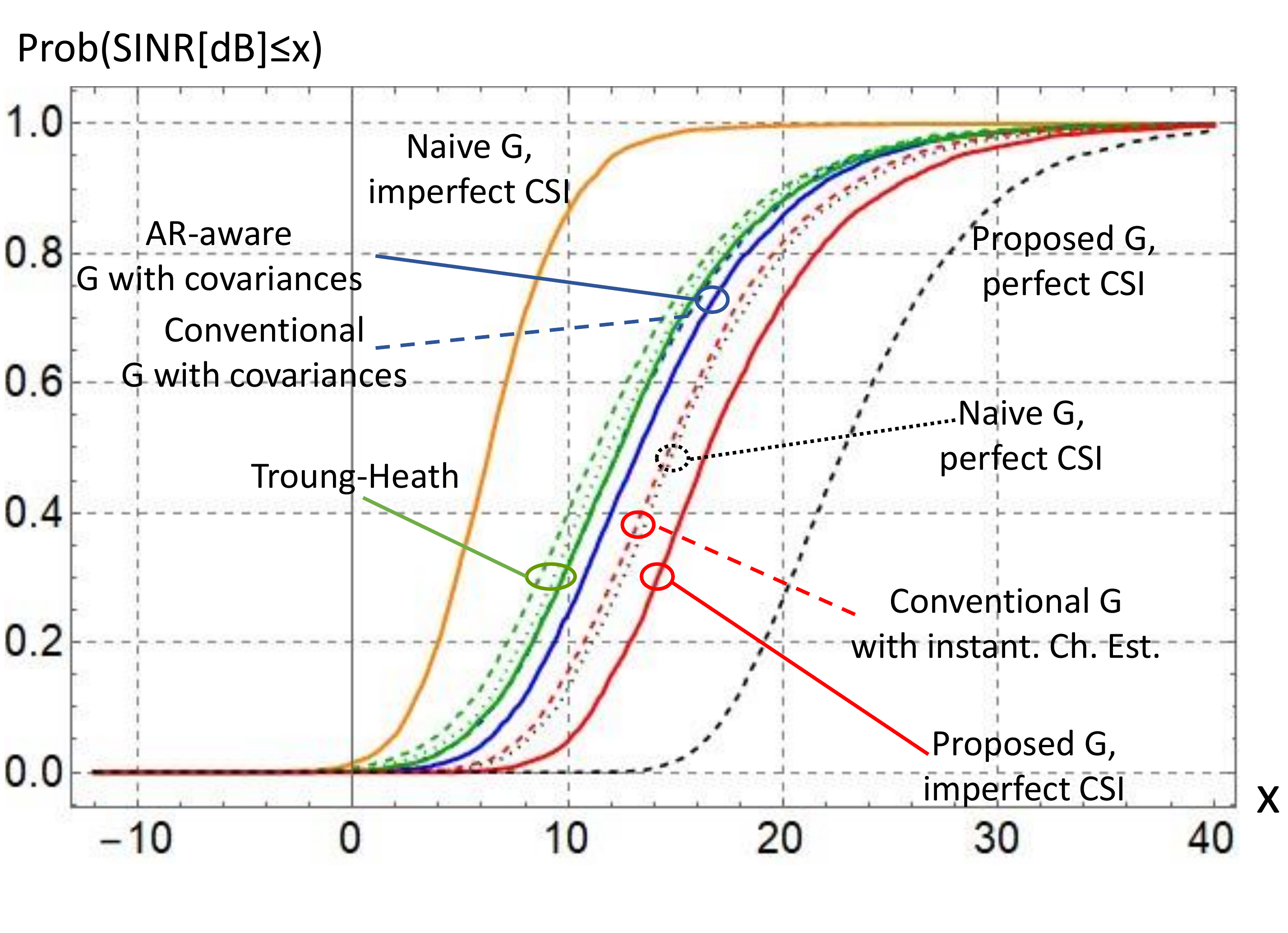}
\caption{\acp{CDF} of the instantaneous \ac{SINR} defined in \eqref{eq:lemma2Eq} when using the proposed
AR-aware MMSE receiver (red solid line) and previously proposed MU-MIMO receivers (see Table \ref{tab:G}).
Note the significant gain
as compared with the AR-aware MU-MIMO receiver that treats interference as noise proposed in \cite{Fodor:2021}
and with Troung and Heath (1), (2), (3) proposed in \cite{Truong:13}.
}
\label{Fig:Fig1}
\end{figure}
\hfill
\begin{figure}[ht]
\centering
\includegraphics[width=1.\columnwidth]{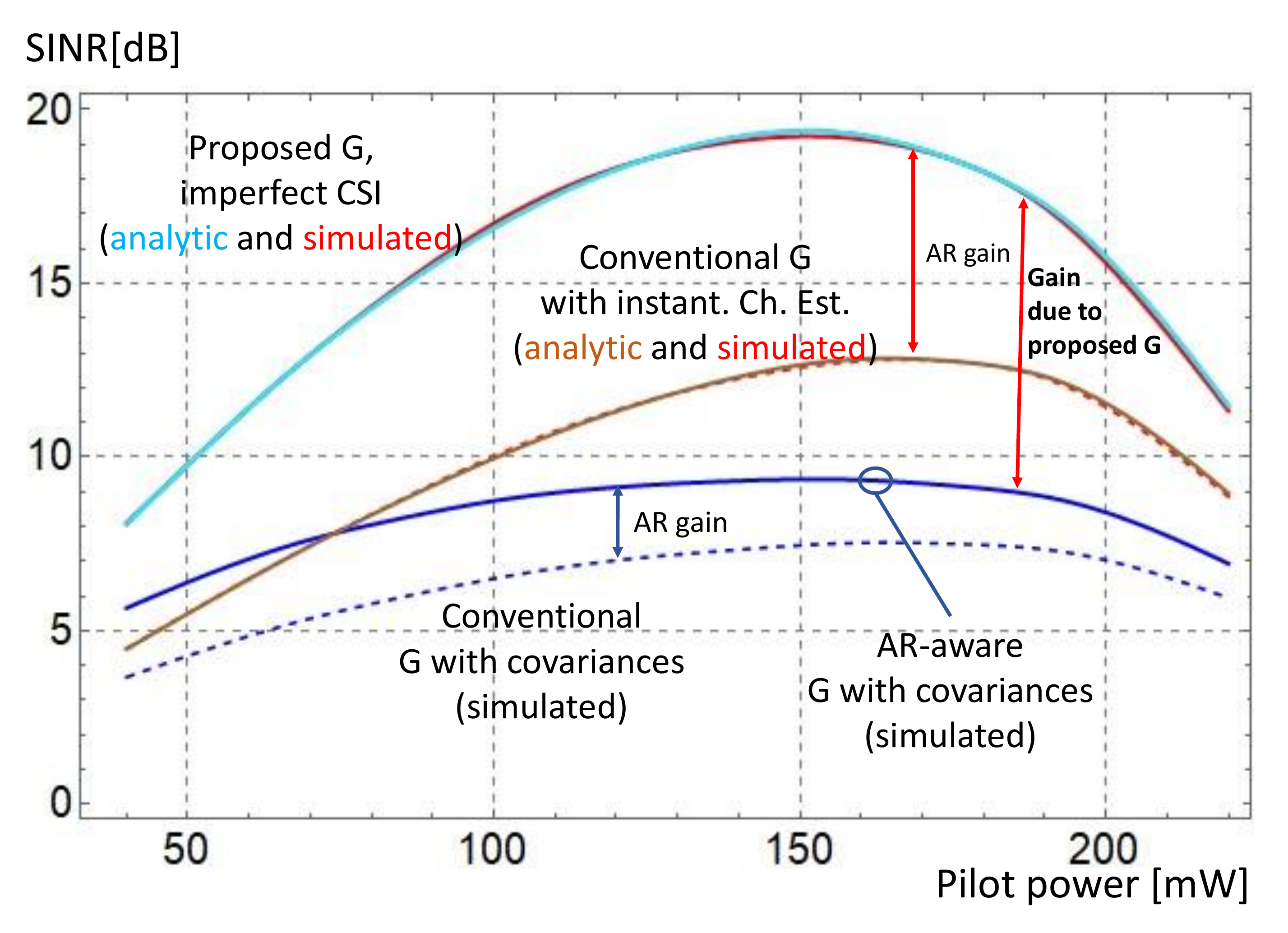}
\caption{Average SINR as a function of the employed pilot power when using the proposed and the state of the art
receivers. The SINR of the proposed receiver is both calculated using Theorem \ref{thm:1} and simulated. Similarly
to the previous figure, we can see the significant gain of the proposed receiver over the receivers developed in
\cite{Fodor:2021} and \cite{Abrardo:19}.}
\label{Fig:Fig2}
\end{figure}

Figure \ref{Fig:Fig1} shows the \ac{CDF} of the \ac{SINR} of the tagged user for the specific case when
the number of users is $K=5$, number of receive antennas at the \ac{BS} is $N_r=100$ and the pilot power
is kept fixed at $P_p=100$ mW. Notice that the proposed receiver, which uses Kalman filter-assisted channel estimation
for all users outperforms the conventional receiver, which does not use Kalman filter for channel estimation.
The potential of the proposed \ac{MMSE} receiver is indicated by the rightmost curve, which shows the \ac{SINR}
performance of this receiver if it has access to perfect channel estimates. Even in the presence of channel
estimation errors, it outperforms all other receivers due to two reasons. First, its structure is modified
as compared with previously proposed receivers and second, it takes advantage of the instantaneous channel
estimates based on multiple observations (i.e. $\mx{\hat h}(t)$ and $\mx{\hat h}(t-1)$).

Figure \ref{Fig:Fig2} shows the average \ac{SINR} performance of the proposed receiver, using Theorem \ref{thm:1},
verified by simulations. The performance of the proposed receiver is compared both with that of the conventional
receiver \cite{Hoydis:2013, Abrardo:19} (termed \ac{MMSE} receiver in those papers), and that of the AR-aware receiver
proposed in \cite{Fodor:2021}, which uses the covariance matrices of the interfering users to suppress \ac{MU-MIMO}
interference. In this Figure, we refer to the gain over the first type of receivers as the "AR gain", since this
gain is due to modified receiver structure, which makes it "AR aware". The gain over the receiver proposed in \cite{Fodor:2021}
is due to estimating all users' channels, rather than treating the \ac{MU-MIMO} interference as noise.
This figure also shows that the analytical \ac{SINR} calculation based on Theorem \ref{thm:1} gives a tight approximation.
\begin{figure}[ht]
\centering
\includegraphics[width=1.\columnwidth]{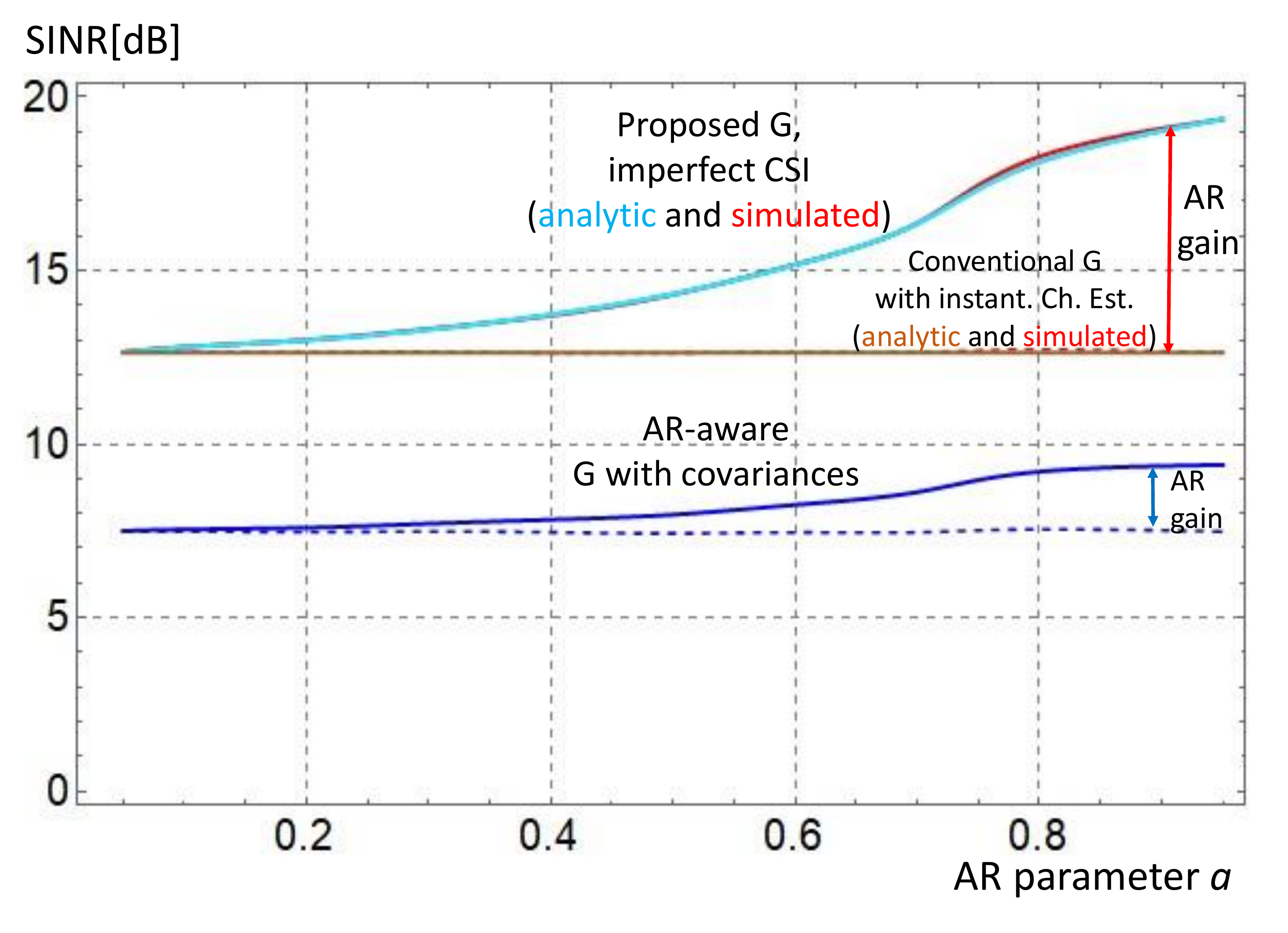}
\caption{Average SINR as a function of the AR parameter  $a$. The proposed receiver falls back
to the receiver that is not AR-aware and uses the instantaneous channel estimates of all users
\cite{Abrardo:19} when $a$ is close to zero. Likewise, the receiver that uses the covariance matrices
of the estimated channels \cite{Fodor:2021} falls back to the conventional receiver \cite{FMT:15} when $a=0$.
}
\label{Fig:Fig3}
\end{figure}
\hfill
\begin{figure}[ht]
\centering
\includegraphics[width=1.\columnwidth]{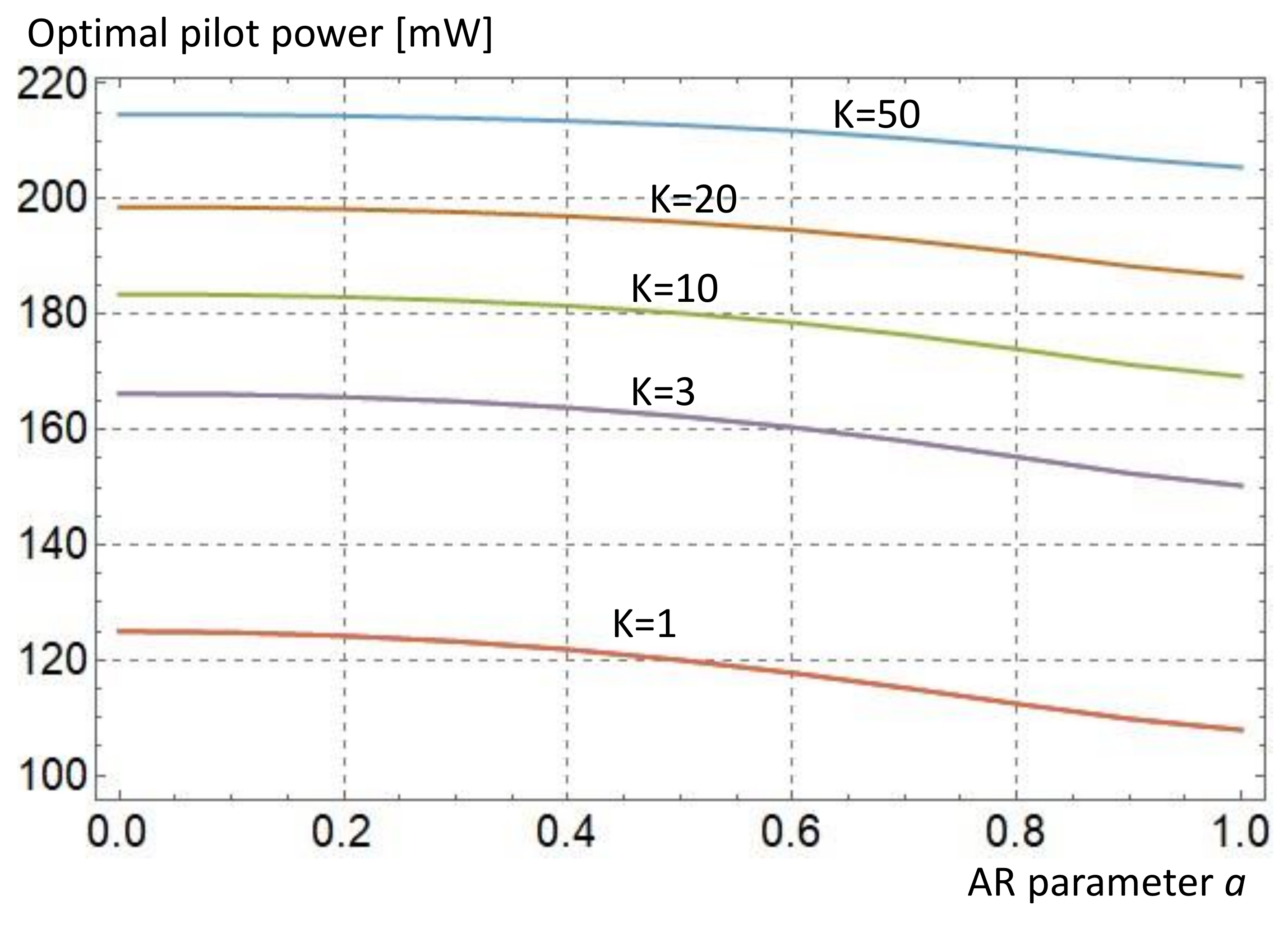}
\caption{Optimum pilot power vs $a$ for $K=1,3,10,20,50$ users. Note that the optimum pilot power
does not depend on the number of antennas. The optimal pilot power increases with increasing number
of users when assuming a total pilot+data power budget.
}
\label{Fig:Fig4}
\end{figure}

Figure \ref{Fig:Fig3} compares the performance of AR-aware receiver developed in \cite{Fodor:2021}
with that of the proposed receiver in this current paper, as a function of the \ac{AR} parameter $a$.
The horizontal lines correspond to the SINR performance of the conventional receivers that do not
exploit the memoryful property of the channel, that is they assume that $a=0$. First, notice that
both receivers take advantage of the AR process of the channel when $a$ is close to 1 ("AR gain").
Second, the currently proposed receiver gains much more by exploiting the channel \ac{AR} process than
the receiver proposed in \cite{Fodor:2021}, since this receiver estimates the channels of all users
rather than treating the interfering users as unknown noise. The sum of these two gains is quite
significant when comparing the \ac{SINR} performance of the conventional \ac{MU-MIMO} receiver by the proposed
\ac{MU-MIMO} receiver when the autocorrelation coefficients of the user channels are high. Such high
autocorrelation property can be achieved in practice by proper pilot symbol allocation in the time
domain.

Figure \ref{Fig:Fig4} shows the optimum pilot power setting as a function of the AR parameter $a$
for systems in which the number of users is $K=1, 3, 10, 20, 50$. This figure assumes that the users
are placed along
a circle around the serving base station, that is, all users have the
same path loss and set their pilot/data power ratio identically.
as mentioned the optimum pilot power
is invariant under of the number of receive antennas ($N_r$) as long as $N_r \geq K$. This figure clearly indicates
that when the number of users is large, each user should increase its pilot power, which implies decreasing
their data power due to the sum pilot and data power constraint. The main reason for this is that while
the pilot signals do not cause interference to each other (due to the assumption on pilot sequence orthogonality),
increasing the number of users increases the MU-MIMO interference level on the received data signals.
Therefore, the optimum pilot allocation in the many users case tends to reduce data power and increase the pilot power levels.
Furthermore, Figure \ref{Fig:Fig4} indicates that the optimum pilot power is decreasing with parameter $a$. An intuitive explanation of this behaviour is that the strong correlation of the channel state in consecutive periods makes easier to acquire the \ac{CSI}.

\begin{figure}[ht]
\centering
\includegraphics[width=1.\columnwidth]{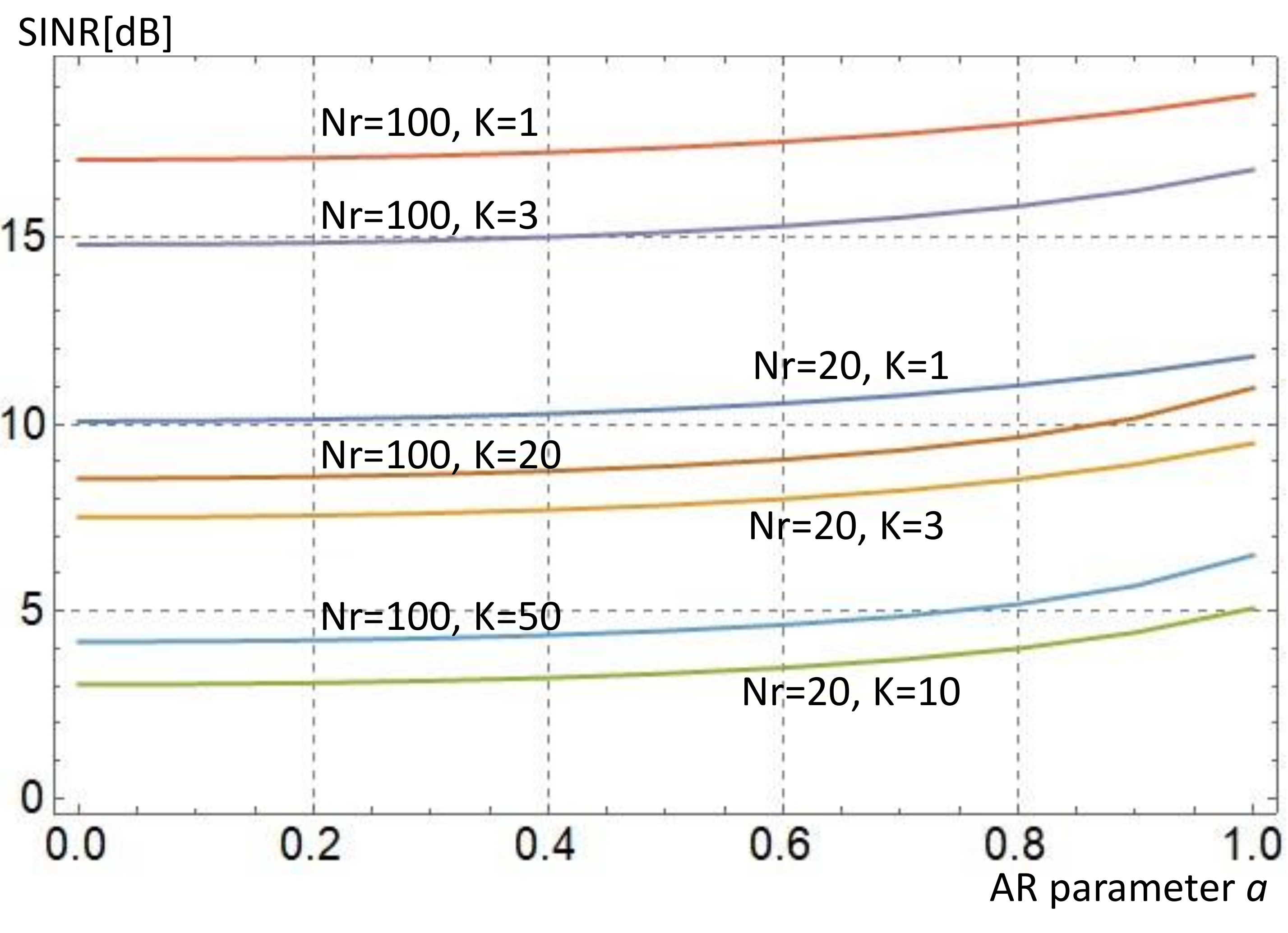}
\caption{\ac{SINR} when using the optimum pilot power vs $a$ for various number of users and antennas.
The achieved \ac{SINR} increases when the AR coefficient is high as compared with the case when the
channel samples are uncorrelated (i.e. block fading) in time.
}
\label{Fig:Fig5}
\end{figure}
\hfill
\begin{figure}[ht]
\centering
\includegraphics[width=1.\columnwidth]{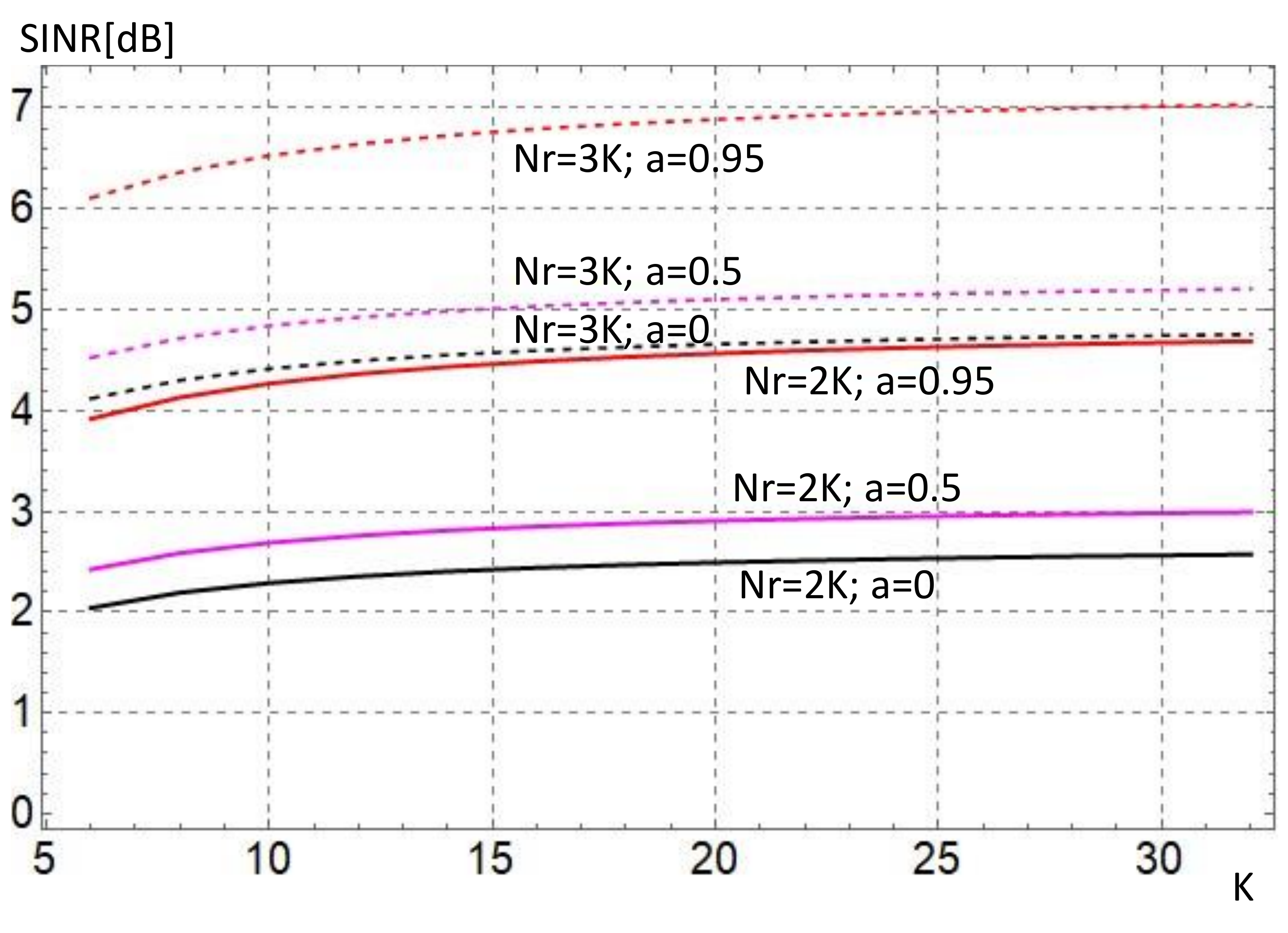}
\caption{\ac{SINR} vs K when $N_r=2K$ and when $N_r=3K$. The \ac{SINR} performance of the $N_r = 2K$ system with $a=0.95$
almost reaches that of the $N_r = 3K$ system with $a=0$.
}
\label{Fig:Fig6}
\end{figure}

Figure \ref{Fig:Fig5} shows the achieved \ac{SINR} when pilot power is set optimally,
as a function of the AR coefficient $a$. Again, we notice that the performance increases as $a$ increases
for all cases. Also, the \ac{SINR} performance of a system with $N_r=100$ and $K=50$ users is somewhat
higher than that of a system with $N_r=20$ and $K=10$. This is expected, since larger number of antennas
implies an improved array gain for all users. We can also see that the gain due to increasing $a$
is similar in all cases.

Figure \ref{Fig:Fig6} uses Theorem \ref{thm:1} to calculate the average \ac{SINR} as a function of
the number of users $K$ when the number of antennas is set to $N_r=2K$ and $N_r=3K$ and when
setting  $a=0$, $a=0.5$ and $a=0.95$. Here we can see that setting $N_r=2K$ with $a=0.95$
gives almost the same \ac{SINR} performance as when having $N_r=3K$ antennas with $a=0$. This result
indicates that when the pilot symbols are sufficiently densely spaced and the autocorrelation
in the channel is well exploited, much lower number of antennas can give a similar \ac{SINR} performance
as that of a system with a high number of antennas.

\begin{figure}[ht]
\centering
\includegraphics[width=\columnwidth]{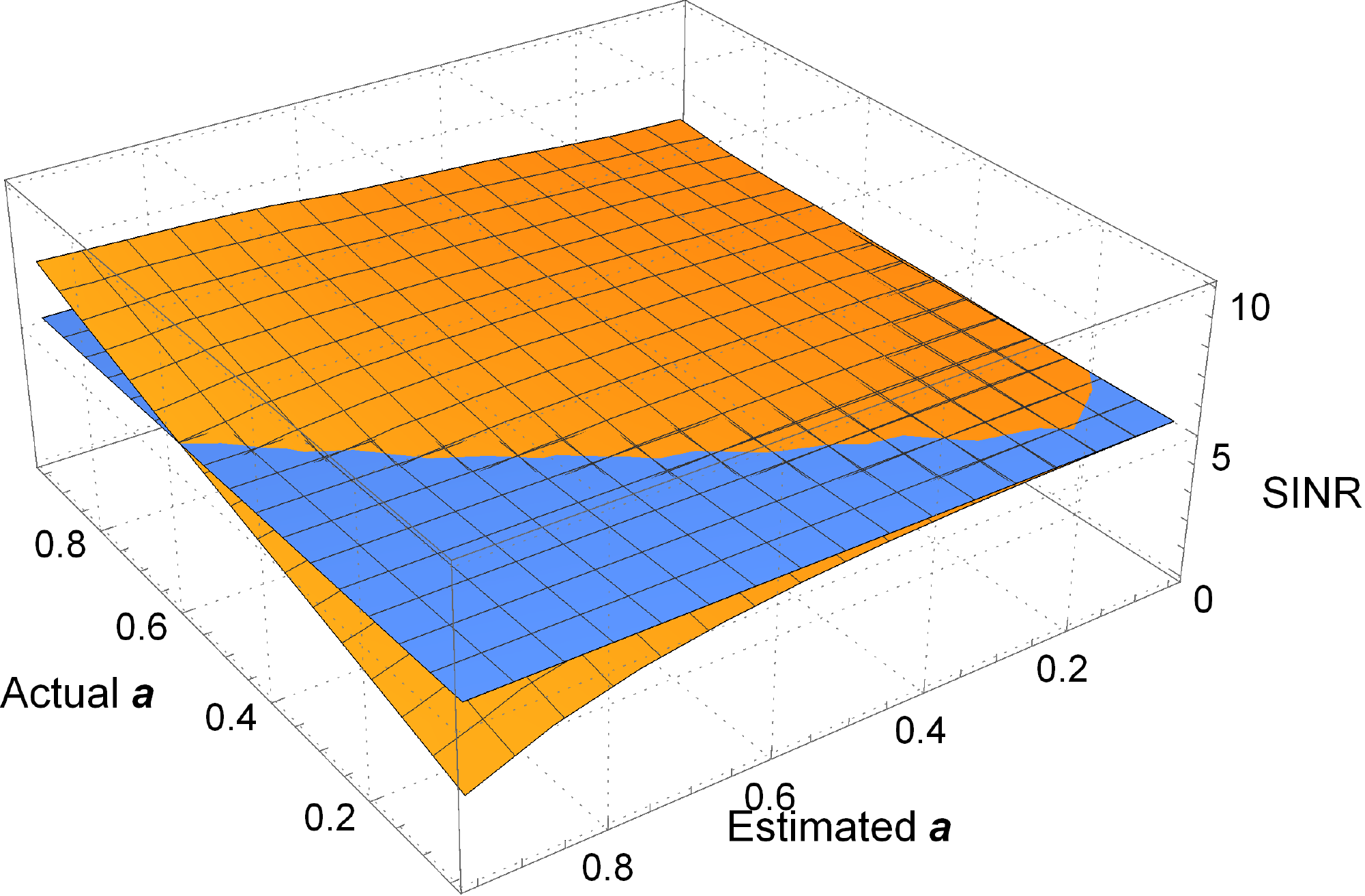}
\vspace{1mm}
\caption{Average \ac{SINR} vs the actual and the estimated AR parameters ($a$ and $\hat a$).
The flat surface indicates
the average SINR performance in a system where $a=0$, which is correctly assumed by the receiver.
}
\label{Fig:Fig7}
\end{figure}

Figure \ref{Fig:Fig7} illustrates the sensitivity of the achieved average \ac{SINR}
when using proposed receiver with
respect to the difference between the estimated and actual $a$ parameters of the \ac{AR}
channel.
The figure shows the actually achieved average \ac{SINR} in a system with $N_r=20$
antennas and $K=5$ users, as a function of the actual ($a$) and estimated ($\hat a$)
AR parameter. The flat surface indicates the \ac{SINR} level that is achieved in a system
with $a=0$ that correctly assumes that $a=0$.

When the actual $a$ is high (greater than 0.8),
the achieved \ac{SINR} is higher than when $a=0$, for all estimated $\hat a$ values. However,
when the actual $a$ is low (the channel is effectively block fading) and the estimated
$\hat a$ is high (the receiver assumes strong correlation in the subsequent channel estimates),
the achieved \ac{SINR} is lower than what is achieved by a conventional receiver. This result
suggests that with proper pilot symbol spacing, when $a$ is high, estimating well the $a$
is also important to fully harvest the gains by using the proposed receiver.

\begin{figure}[ht]
\centering
\includegraphics[width=1.\columnwidth]{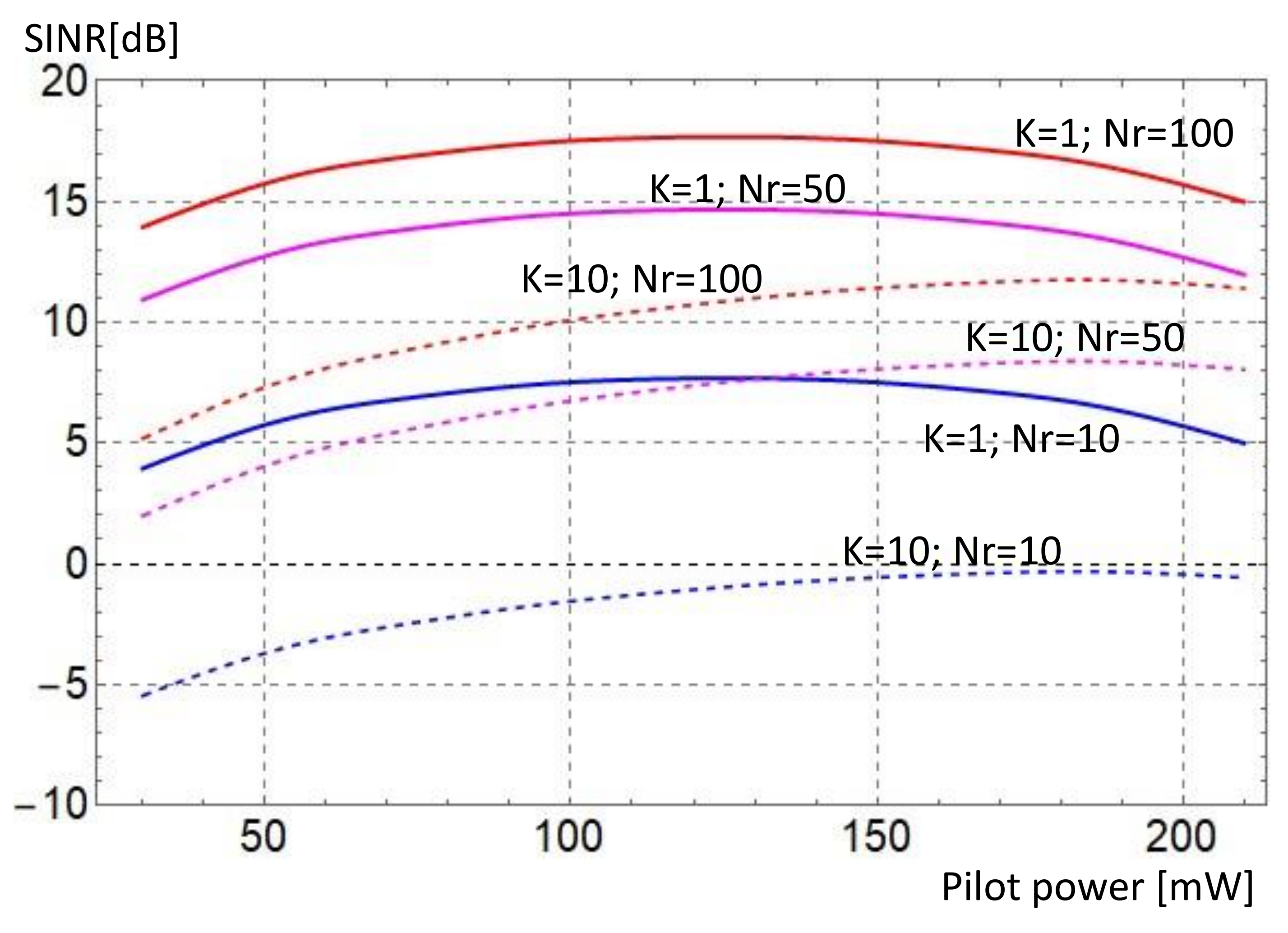}
\caption{SINR performance calculated analytically using Theorem 2 of the proposed AR-aware $\mx{G}^\star$ receiver as a function of the pilot power in different scenarios
in terms of number of users $K$ (i.e. single user or $K=10$) and number of antennas at the \ac{BS}, (i.e. $N_r=10, 50, 100$) at $a=0$.}
\label{Fig:Fig8a0}
\end{figure}
\hfill
\begin{figure}[ht]
\centering
\includegraphics[width=1.\columnwidth]{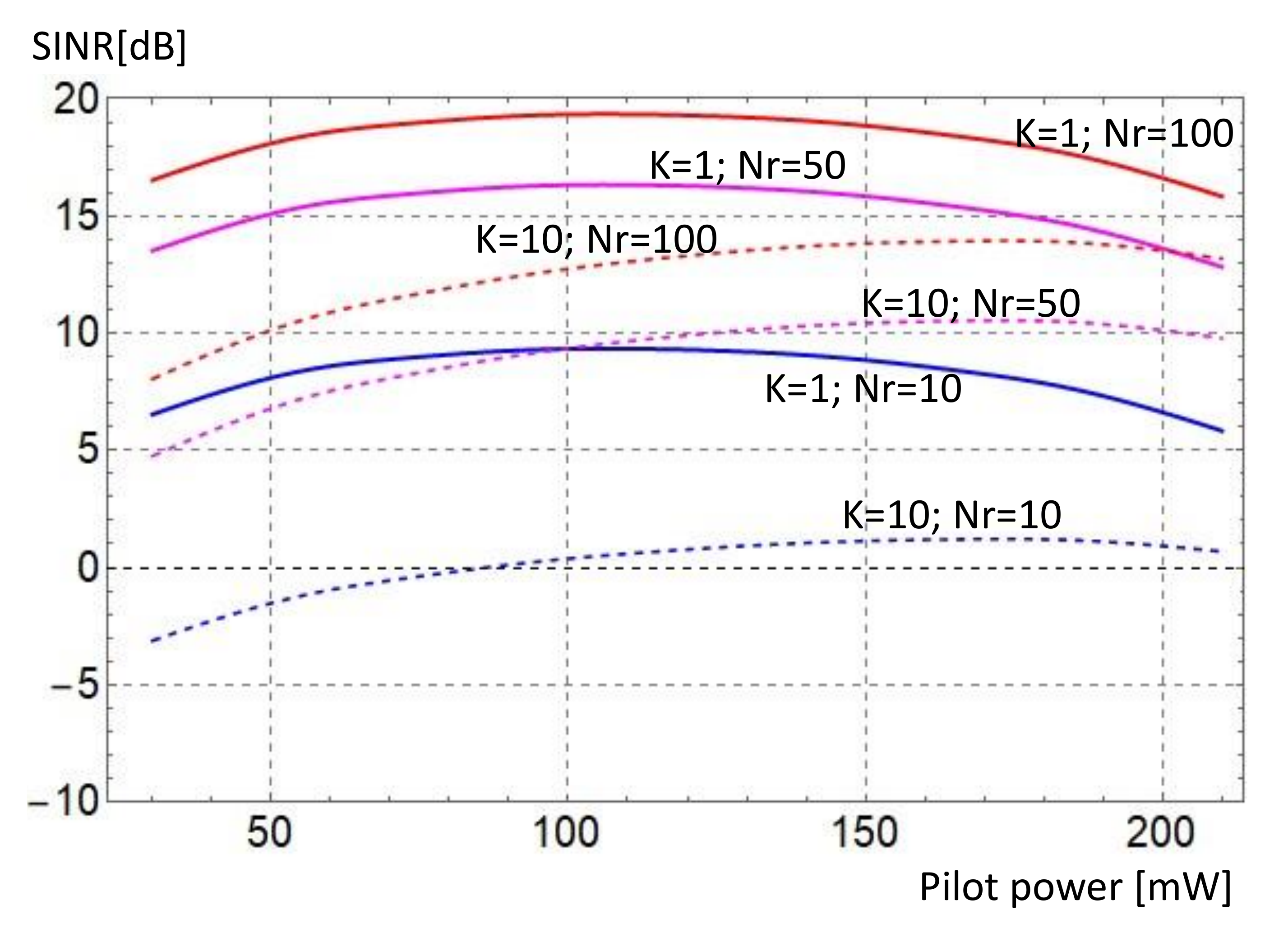}
\caption{SINR performance calculated analytically using Theorem 2 of the proposed AR-aware $\mx{G}^\star$ receiver as a function of the pilot power in different scenarios
in terms of number of users $K$ (i.e. single user or $K=10$) and number of antennas at the \ac{BS}, (i.e. $N_r=10, 50, 100$) at $a=0.95$.}
\label{Fig:Fig8a095}
\end{figure}

\rev{
Finally, Figures \ref{Fig:Fig8a0} and \ref{Fig:Fig8a095} compare the average \ac{SINR} performance of single and multiuser ($K=10$) systems when $a=0$ and when $a=0.95$
when the number of base station antennas is low $N_r=10$ and high $N_r=50$ or $N_r=100$. Notice that in the case of a memoryful MIMO channel ($a=0.95$) properly
setting the pilot power and exploiting the memoryful property of the channel, an average SINR above 0 dB can be achieved even with a relatively low number of antennas
(see the case of $N_r=K=10$), whereas in the case of $a=0$ the average SINR stays below 0 dB, especially if the pilot power is not properly tuned.
}

\section{Conclusions and Outlook}
\label{Sec:Conc}

In this paper we proposed a new \ac{MU-MIMO} receiver, whose distinguishing features
are its capability to utilize the instantaneous channel estimate of each user,
and to exploit the memoryful property of the \ac{MU-MIMO} wireless channels (AW-awareness) when these channels
evolve according the an \ac{AR} process. The main contribution of this paper is the
new \ac{MU-MIMO} receiver structure (Proposition \ref{P2}) and its performance analysis
facilitated by Proposition \ref{prop:Hoydis} and Theorem \ref{thm:1}. This receiver
and its performance analysis extends the results by \cite{Hoydis:2013}
in the sense that
(1) the proposed receiver exploits the memoryful property of the AR channels rather than
treating them as block fading and (2) due to Theorem \ref{thm:1} it allows the calculate
the average \ac{SINR} without solving a system of fixed point equations.
Our numerical results
indicate that the proposed receiver outperforms previously proposed \ac{MU-MIMO} receivers.
An important future work, which is outside the scope of the present paper, is to find
the optimal pilot power levels when the users are randomly placed in the coverage area of the cell,
and, consequently, have different large scale fading parameters.
\rev{Also, in the light of the results by multicell \ac{MU-MIMO} receivers studied
by \cite{Bjornson:18}, \cite{Boukhedimi:18} and \cite{Sanguinetti:19}
in block fading
environments, it is an exciting question, whether the proposed receiver in this paper can be
extended to multicell systems.}

\appendices

\section{Proof of Lemma \ref{lem:mmsechannel}}

\rev{
\begin{proof}
The MMSE channel estimator aims at minimizing the MSE between the channel estimate
$\mathbf{\hat h}_{\textrm{MMSE}}(t) = \mathbf{H}^\star \mathbf{\hat Y}^p(t)$ and the channel $\mathbf{h}(t)$, where
$\mathbf{H} \in \mathds{C}^{N_r \times 2 \tau_p N_r }$,
$\mathbf{\hat Y}^p(t)=\begin{bmatrix}
\mathbf{\tilde Y}^p(t) \\
\mathbf{\tilde Y}^p(t-1)
\end{bmatrix} \in \mathds{C}^{2 \tau_p N_r \times 1}$ and
$\mathbf{H^\star}= \text{arg} \min_{\mathbf{H}} \mathds{E}_{\mathbf{h},\mathbf{n}}\{ ||\mathbf{H} \mathbf{\hat Y}^p(t) - \mathbf{h}(t)||_F^2 \}$.
The solution of this quadratic optimization problem is
$\mathbf{H^\star}= \mx{b}^H \mx{F}^{-1}$
with
\begin{align*}
\mx{F} &= \mathds{E}_{\mathbf{h},\mathbf{n}}\left(\mx{\hat Y}^p \mbox{$\mx{\hat Y}^p$}^H\right) \nonumber \\
&=
\begin{bmatrix}
\alpha^2P_p \mx{S} \mx{C} \mx{S}^H+ \sigma_p^2 \mathbf{I}_{N_r\tau_p} &
\alpha^2P_p \mx{S} (\mx{A}\mx{C}) \mx{S}^H\\
\alpha^2P_p \mx{S} (\mx{C}\mx{A}^{H}) \mx{S}^H &
\alpha^2P_p \mx{S} \mx{C} \mx{S}^H+ \sigma_p^2 \mathbf{I}_{N_r\tau_p}
\end{bmatrix}
\\&=
\alpha^2P_p~
\left(\mx{s} \mx{s}^H \otimes \mx{M} \right) + \sigma_p^2 \mathbf{I}_{2 N_r\tau_p},
\\
\mx{b} &= \mathds{E}_{\mathbf{h},\mathbf{n}}\left( \mx{\hat Y}^p \mx{h}^{H}(t)\right)
=
\begin{bmatrix}
\alpha\sqrt{P_p} \mx{S} \mx{C} \\
\alpha\sqrt{P_p} \mx{S} (\mx{C}\mx{A}^{H})
\end{bmatrix} \nonumber \\
&=
\alpha\sqrt{P_p} ~~\left(\mx{s}\otimes
\begin{bmatrix}
\mx{C} \\
\mx{C} \mx{A}^{H}
\end{bmatrix}\right),
\end{align*}
where we utilized $\mathbf{S} \triangleq \mathbf{s}\otimes \mathbf{I}_{N_r}$ and $\mathbf{S}^H \mathbf{\tilde N}(t)=\mathbf{s}^H \mathbf{N}(t)$.
That is
\begin{align*}
\mx{H^\star}&=\mx{b}^H \mx{F}^{-1} \nonumber \\
&= \frac{1}{\alpha\sqrt{P_p} \tau_p}
\begin{bmatrix}
\mx{C} &
\mx{A}\mx{C}
\end{bmatrix} \left( \frac{\sigma_p^2}{\alpha^2P_p \tau_p} \mx{I}_{2N_r} +  \mx{M}\right)^{-1}
\left(\mx{s}^H\otimes\mx{I}\right).
\end{align*}
The \ac{MMSE} estimate is then expressed as
\begin{align}
\label{eq:MMSEt}
&\mathbf{\hat h}_{\textrm{MMSE}}(t)=\mathbf{H^\star} \mathbf{\hat Y}^p(t) = 
\begin{bmatrix}
\mx{C} &
\mx{A} \mx{C}
\end{bmatrix}
\left( \frac{\sigma_p^2}{\alpha^2P_p \tau_p} \mx{I}_{2N_r} +  \mx{M}\right)^{-1} \nonumber \\
&~~~~~~~~~~~~~~
.
\begin{bmatrix}
\mathbf{h}(t) + \frac{1}{\alpha\sqrt{P_p} \tau_p} \mathbf{s}^H \mathbf{N}(t) \\
\mathbf{h}(t\!-\!1) + \frac{1}{\alpha\sqrt{P_p} \tau_p} \mathbf{s}^H \mathbf{N}(t\!-\!1)
\end{bmatrix},
\end{align}
which gives the lemma.
\end{proof}
}

\section{Proof or Proposition \ref{prop:Hoydis}}
\label{App:Hoydis}
Starting from \eqref{eq:gamma},
we first apply \cite[Lemma 1, eq. (47)]{Truong:13}
which states that, if
$\left(\mx{B}\mx{B}^H + \boldsymbol{\beta} \right)^{-1}$
has a uniformly bounded spectral norm, then
\begin{align}
&\frac{1}{N_r}
\mx{b}^H \left(\mx{B}\mx{B}^H + \boldsymbol{\beta} \right)^{-1} \mx{b}
-
\frac{1}{N_r} \text{tr}\left(\mx{\Phi}
\left(\mx{B}\mx{B}^H + \boldsymbol{\beta} \right)^{-1}\right) \xrightarrow[N_r\rightarrow\infty]{\text{a.s.}} 0.
\end{align}
In the second step we apply \cite[Theorem 1]{Hoydis:2013}, which states that,
if $N_r\to\infty$ and $\limsup_{N_r\to\infty} K/N_r<\infty$, then
\begin{align}
\label{eq:asc1}
&\frac{1}{N_r}\text{tr}\left(\mx{\Phi} \left(\mx{B}\mx{B}^H + \boldsymbol{\beta} \right)^{-1}\right) -
\frac{1}{N_r}\text{tr} \Big(\mx{\Phi} \mx{T}\Big) \xrightarrow[N_r\rightarrow\infty]{\text{a.s.}} 0,
\nonumber \\
&\text{~where~~} 
~\mx{T} \triangleq \left(\frac{1}{N_r} \sum_{k=2}^K
\frac{ \mx{\Phi}_k }{1+\delta_{k}}
+ \boldsymbol{\beta}\right)^{-1},
\end{align}
and $\delta_{k}$, for $k=2,\ldots,K$ are the solution of
\begin{align}
\label{eq:asc2}
\delta_{k} &=
\frac{1}{N_r}
\text{tr}\left( \mx{\Phi}_k \left(\frac{1}{N_r} \sum_{\ell=2}^K \frac{\mx{\Phi}_\ell}{1+\delta_{\ell}}
+\boldsymbol{\beta} \right)^{-1}\right).
\end{align}
Adding equations \eqref{eq:asc1} and \eqref{eq:asc2} we get that
\begin{align}
    \mx{b}^H \left(\mx{B}\mx{B}^H + \boldsymbol{\beta} \right)^{-1} \mx{b} -
    \frac{1}{N_r}\text{tr} \Big(\mx{\Phi} \mx{T}\Big) \xrightarrow[N_r\rightarrow\infty]{\text{a.s.}} 0,
\end{align}
which together with equation \eqref{eq:gamma} gives the desired result.



\section{Proof of Theorem \ref{thm:2}}
\label{Sec:AppVI}
To prove Theorem \ref{thm:2}, we need the following Lemma \rev{regarding the moments of the random variable $\omega_n$:}
\begin{lem}
\label{lem:4}
\rev{Let $\omega_n$ and $\bar{\lambda}$ be defined as in Theorem \ref{thm:2},
we can then state the following relationship between the moments of $\omega_n$ and the powers of $\bar{\lambda}$,}
\begin{align}
\lim_{n\rightarrow \infty} \frac{\mathds{E}\{ \omega_n^r\}}{n^{r-1}} = \bar{\lambda}^r.
\end{align}
\end{lem}

\begin{proof}[Proof of Lemma \ref{lem:4}]
The random matrix $\mx{v}^{(n)}\left(\mx{v}^{(n)}\right)^H$ is rank one and thus has $n-1$ eigenvalues equal to 0
and one eigenvalue equal to $\text{tr}\left(  \mx{v}^{(n)}\left(\mx{v}^{(n)}\right)^H  \right) = \sum_{i=1}^{n} \|v^{(n)}_i\|^2$.
Note that $Y_i \triangleq \|v^{(n)}_i\|^2$ has an exponential distribution with mean $\lambda_i$.
Since $\omega_n$ is one of these eigenvalues, randomly selected, we have
\begin{align}
\lim_{n\rightarrow \infty} \frac{\mathds{E}\{ \omega_n^r\}}{n^{r-1}} &= \lim_{n\rightarrow \infty}
\frac{\frac{1}{n}\mathds{E}\{\left(\sum_{i=1}^{n} Y_i \right)^r\}}{n^{r-1}} \nonumber\\
&=
\lim_{n\rightarrow \infty} \mathds{E}\left\{ \left( \frac{\sum_{i=1}^{n} Y_i}{n} \right)^r \right\}. \label{eq:y_lim}
\end{align}
Furthermore by the strong law of large numbers as $n\to\infty$
\begin{align}
\label{eq:slln}
    \frac{\sum_{i=1}^{n} Y_i}{n}  \xrightarrow[n\rightarrow\infty]{\text{a.s.}} \bar{\lambda}
    &\Rightarrow \left(\frac{\sum_{i=1}^{n} Y_i}{n}\right)^r \xrightarrow[n\rightarrow\infty]{\text{a.s.}} \bar{\lambda}^r \nonumber\\
    &
    \Rightarrow \lim_{n\rightarrow \infty} \mathds{E}\left\{ \left( \frac{\sum_{i=1}^{n} Y_i}{n} \right)^r \right\} = \bar{\lambda}^r.
\end{align}
Equations \eqref{eq:y_lim} and \eqref{eq:slln} give the Lemma.
\end{proof}

For the proof of the Theorem \ref{thm:2}, \rev{in addition to Lemma \ref{lem:4}},
we will use the equivalent (cf. \eqref{eq:defr}) definition of the $\mathcal{R}$-transform of
a random variable $X$ using its cumulants \cite{Muller:13}:
\begin{align}
\mathcal{R}_X(s) \triangleq \sum_{k=0}^{\infty}\kappa_{k+1}s^{k},
\end{align}
where $\kappa_k$ is the k'th cumulant of $X$, that is
\begin{equation}\label{cumulant}
    \kappa _{k}= \left. \frac{d^k}{ds^k}  K_X(s)\right|_{s=0} = K_X^{(k)}(0),
\end{equation}
and $K_X(s)$ is the cumulant generating function
    $K_X(s) \triangleq
    \operatorname \log \mathds{E} \left[e^{sX}\right]$.
\rev{We use this definition in the proof as it is often useful to have two equivalent definitions of a function,
and use one of them to say something about the other.
In this case we use the cumulant definition of the $\mathcal{R}$-transform to
be able to state results about the Stieltjes transform.}
\rev{In order to calculate the cumulants $\kappa_k$ we calculate the value of the derivatives of $K_X(s)$ at $s=0$,
we do this through the derivatives of the moment generating function $M_X(s) \triangleq \mathds{E} \left[e^{sX}\right]$ of $X$ as follows.}
First define $m_i(s) \triangleq M_X^{(i)}(s)/M_X(s)$,
and define the order of the product
\begin{align}
    \prod_{k=1}^K m_{i_k}^{j_k}(s) ~~~\text{to be}~~~ \sum_{k=1}^K i_k j_k.
\end{align}
Notice that by the quotient rule
\begin{align}
\frac{d}{ds}m_i(s) &= \frac{d}{ds}\frac{M_X^{(i)}(s)}{M_X(s)} \nonumber \\
&=
\frac{  M_X^{(i+1)}(s)  M_X(s) - M_X^{(i)}(s)M_X^\prime(s)}{M_X^2(s)} \nonumber \\
&=
m_{i+1}(s) - m_i(s)m_1(s).
\end{align}
Thus, by the product rule the derivative of an order $k$ product is a sum of order $k+1$ products, and so the $k$'th cumulant
\begin{align}
\kappa_k &= \left. \frac{d^k}{ds^k}  K_X(s)\right|_{s=0} = \left. \frac{d^{k-1}}{ds^{k-1}} m_1(s) \right|_{s=0},
\end{align}
is a sum of order $k$ products at $s=0$, and one of the terms of this sum is $m_k(0)$.
Furthermore, by the definition of $m_k$, we have $m_k(0) = \mathds{E}\{X^k\}$.

More specifically, looking at the random variable $\omega_n$, we know from Lemma \ref{lem:4} that
$\mathds{E}\{\omega_n^k\}$ and hence $m_k(0)$ is $O(n^{k-1})$.
Consequently,
any order $k$ product at $s=0$ other than $m_k(0)$ is $O(n^{k-2})$, and so by Lemma \ref{lem:4}:
\begin{align}
\lim_{n\rightarrow \infty} \frac{\kappa_k}{n^{k-1}} &= \bar{\lambda}^k.
\end{align}
We can now derive \eqref{eq:Th1}:
\begin{align}
\lim_{n\rightarrow \infty} & \mathcal{R}_{\omega_n}\left( \frac{s}{n} \right)
= \lim_{n\rightarrow \infty} \sum_{k=0}^{\infty}\kappa_{k+1}\left(\frac{s}{n}\right)^{k} \nonumber\\
&=
\lim_{n\rightarrow \infty} \sum_{k=0}^{\infty}\frac{\kappa_{k+1}}{n^k}s^{k} = \sum_{k=0}^{\infty} \bar{\lambda}^{k+1}s^k 
= \frac{\bar{\lambda}}{1 - s\bar{\lambda}},
\end{align}
which completes the proof.

\section{Proof of Theorem \ref{thm:1} Using the Stieltjes and $\mathcal{R}$-Transforms}
\label{Sec:AppV}
The first proof of Theorem \ref{thm:1} relies on random matrix theory using
the Stieltjes transform, the $\mathcal{R}$-transform and
\rev{Corollary \ref{cor:rtrafo}} of Theorem \ref{thm:2}.
To determine \eqref{eq:averageSINR},
we use the \rev{spectral} decomposition of the \rev{Hermitian matrix}
and define
$\mathbf{y}\triangleq \mathbf{U} \mathbf{b}$.
Accordingly, \eqref{eq:averageSINR} becomes
\begin{equation}\nonumber
\begin{aligned}
\bar{\gamma}
&=
\mathds{E}_{\mathbf{y},\lambda_i,i=1\ldots N_r}
\left\{\mx{y}^H\mathbf{U}\mathbf{U}^H (\mathbf{\Lambda}+ \beta \mathbf{I}_{N_{r}})^{-1} \mathbf{U}\mathbf{U}^{H}\mx{y}\right\} \nonumber \\
&=
\mathds{E}_{\mathbf{y},\lambda_i,i=1\ldots N_r}\left\{\sum_{i=1}^{N_r}  \frac{|y_i|^2}{\lambda_i + \beta}\right\},
\end{aligned}
\end{equation}
where
$y_i$ is $i$th element of the vector
$\mathbf{y}$ and $\lambda_{i}$ is the $i$th eigenvalue of $\mx{B}\mx{B}^H$.

Since the $\mx{U}$ matrix is unitary, $\mathbf{y}$ and $\mathbf{b}$,
have same distribution, i.e. 
$\mathbf{y}\sim\mathcal{CN}(0,\phi \mathbf{I}_{N_{r}})$
and
$\mathds{E}\left\{|y_i|^2\right\}=\phi;~i=1 \dots N_r$, where recall that $\phi=\phi_1$ (tagged user).
Moreover, since the interference matrix $\mx{B}\mx{B}^H$ is independent of
$\mathbf{b}$,
$\mathbf{y}$ is independent of the eigenvalues $\lambda_i$, and hence
\begin{equation}
\label{eq:gamma25}
\bar{\gamma}=\phi \cdot
\mathds{E}_{\lambda_i,i=1\ldots N_r} \left(\sum_{i=1}^{N_r}  \frac{1}{\lambda_i + \beta}\right).
\end{equation}
Assuming that $N_r, K \rightarrow \infty$, with $K/N_r$ fixed, and using equations
(13) and (14) of \cite{Livan:11} we obtain:
\begin{align}
\mathds{E}_{\lambda_i,i=1\ldots N_r} \left\{\sum_{i=1}^{N_r}  \frac{1}{\lambda_i + \beta}\right\} &=
N_r \mathds{E}_{\lambda} \left\{\frac{1}{\lambda + \beta}\right\},
\end{align}
where $\lambda$ is a randomly selected eigenvalue out of the spectrum of $\mx{B}\mx{B}^H$.

\rev{A first key observation is that the
Stieltjes transform of the distribution of $\lambda$ at $s=-\beta$ is closely related to $\bar{\gamma}$:}
\begin{align}
\label{g-gamma}
G_{\rev{\lambda}}(-\beta) &\overset{(a)}{=}
\int_x  \frac{1}{x+\beta} d P_\lambda(x) \overset{(b)}{=}
\mathds{E}_{\lambda} \left\{\frac{1}{\lambda + \beta}\right\} \overset{(c)}{=}
\frac{\bar{\gamma}}{N_r \phi},
\end{align}
\rev{where $(a)$ is due to definition of the Stieltjes transform, and $(b)$ is due to
noticing that the left hand side of $(b)$ is by definition the expectation of $1/(\lambda+\beta)$.}
Finally,  in the last equation we used \eqref{eq:gamma25}.
\rev{This implies that if we can find an appropriate $\beta$ for which it holds that:}
\rev{
\begin{align}
\label{eq:Gbetaw}
G_\lambda\left(-\beta\right) = w,
\end{align}
where $w \triangleq \frac{\bar{\gamma}}{N_r \phi}$, then according to \eqref{g-gamma}
we found $\bar{\gamma}$ in the form of:
$N_r \phi G_\lambda\left(-\beta\right) = \bar{\gamma}$.
}
\rev{To find such a $\beta$, recall} that for the Hermitian matrix associated with the tagged user
$\mx{B}\mx{B}^H =\sum_{k=2}^K \mathbf{b}_k \mathbf{b}_k^H$
with
\begin{align}
\mathbf{b}_k &\sim \mathcal{CN}(0,\phi_{k}\mathbf{I}_{N_{r}}).
\end{align}
\rev{Furthermore,} we will utilize the following identity (see \eqref{eq:defr}):
\begin{align}
G_\lambda\left(\mathcal{R}_\lambda(-w)-\frac{1}{w}\right)&=w.
\label{grtr}
\end{align}

Furthermore, assuming that $N_r \rightarrow \infty$,
the family of matrices $\mathbf{b}_k \mathbf{b}_k^H$ ($k=1, \dots, K$)
is almost surely asymptotically free \cite{Muller:13}.
Consequently, the $\mathcal{R}$-transform of the sum of matrices $\mathbf{b}_k \mathbf{b}_k^H$ equals the sum of their individual $\mathcal{R}$-transforms.

Recall that by Corollary \ref{cor:rtrafo}, the $\mathcal{R}$-transform of
\rev{a randomly selected eigenvalue $\omega$ of}
$\mathbf{b}_k \mathbf{b}_k^H$ is
$\rev{R_{\omega}(w)}
\approx \frac{\phi_k}{1-N_r \phi_k w}$. 
Hence, utilizing the additive property of the $\mathcal{R}$-transform, for a randomly selected eigenvalue $\Omega$ of $\mx{B}\mx{B}^H$ we get:
\begin{align}
\label{rtr}
\rev{\mathcal{R}_{\Omega}(w)}&
= \sum_{k=2}^K \frac{\phi_k}{1-N_r \phi_k w}.
\end{align}
Substituting \eqref{rtr} into \eqref{grtr}
we have:
\begin{align}
G\left(\sum_{k=2}^K \frac{\phi_k}{1\rev{+}N_r \phi_k w}-\frac{1}{w}\right)=w,
\label{grtr1}
\end{align}
for all $w > 0$.
From this equation it is also evident that the expression inside the $G$-transform is injective for $w > 0$.
Comparing 
\rev{\eqref{eq:Gbetaw} and \eqref{grtr}}, we have that:
\rev{
\begin{align}
-\beta &= \mathcal{R}_\lambda(-w)-\frac{1}{w}
\end{align}}
with
$w=\frac{\bar{\gamma}}{N_r \phi}$, from which, \rev{using \eqref{grtr1},} it follows that
$\bar{\gamma}$ satisfies the equation:
\begin{align}
\beta
&=
\left. \frac{1}{w}-\sum_{k=2}^K \frac{\phi_k}{1+N_r \phi_k w}\right|_{w=\frac{\bar{\gamma}}{N_r \phi}},
\end{align}
which is equivalent with \eqref{eq:SINR35}.
\rev{It is important to note that} there cannot be more than one value of $\bar{\gamma}$ that satisfies the equation above since the RHS is injective in $w$.

\section{Proof or Theorem \ref{thm:1} Using the Trace Approximation}
\label{Sec:AppVII}
To prove Theorem \ref{thm:1}, we first notice that
in the special case of diagonal covariances with equal elements,
we have that (see \eqref{eq:Phi}): 
\begin{align}
\bs{\Phi} &= \phi \mx{I}_{N_r} = \alpha^2 P \left(\hat{e} c + \check{e} c a^*\right) \mx{I}_{N_r}.
\end{align}
In this special case, \rev{i.e.\ when $\bs{\Phi}$ is diagonal with equal diagonal elements}, from \eqref{eq:gammaT}
it follows that for the tagged user (User-$1$), it holds that:
\begin{align}
\label{eq:gammaapprox}
\bar \gamma & \approx \phi \cdot \text{tr}\left(\mx{T}\right).
\end{align}

Also, in this case, the definition of $\mx{T}$ in \eqref{eq:Tdef}
simplifies to:

{\small
\begin{align}
\label{eq:Tell}
\mx{T} & \triangleq
\bigg(\frac{1}{N_r}\sum_{j=2}^K \frac{\phi_j}{1+\delta_j} \mx{I}_{N_r}
+ \underbrace{\sum_{k=1}^K \alpha_k^2 P_k z_k \mx{I}_{N_r}
+ \sigma_d^2 \mx{I}_{N_r}}_{\triangleq \beta\mx{I}_{N_r}}\bigg)^{-1},
\end{align}
}

\noindent where, \rev{according to \cite{Hoydis:2013} and \cite{Wagner:2012}}, the $\delta_{j}$:s satisfy:
\begin{align}
\label{eq:deltak2}
\delta_k &= \phi_k
\cdot \text{tr}\left(\left(\frac{1}{N_r}\sum_{j=2}^K \frac{\phi_j}{1+\delta_{j}}+\beta\right)^{-1} \mx{I}_{N_r}\right);~~
k=2 \dots K.
\end{align}
Comparing \eqref{eq:gammaapprox}, \eqref{eq:Tell} and \eqref{eq:deltak2}, we notice that:
$\delta_k = \phi_k \cdot \frac{\bar \gamma}{\phi} ~\forall k \ne 1$. 
Substituting this into \eqref{eq:deltak2}, we obtain:
\begin{align}
\bar\gamma &=
N_r \phi \left(\sum_{j =2}^K
\frac{\phi_j}{1+\frac{\phi_j}{\phi} \bar\gamma}+\beta\right)^{-1}.
\end{align}
From this equation we get
$\beta = \frac{N_r \phi}{\bar\gamma} -
\sum_{j =2}^K \frac{\phi_j}{1+\frac{\phi_j}{\phi} \bar\gamma},
$
which is identical with \eqref{eq:SINR35}.

\section{Proof of Proposition \ref{prop:OptP3}}

\begin{proof}
First notice that
substituting $\phi = \alpha^2 P (\hat{e}c + \check{e}ca)$,
$\beta = K\alpha^2z + \sigma_d^2$, and $z=c-(\hat{e}c + \check{e}ca)$, the optimization problem in \eqref{eq:phiperbeta}
can be rewritten as:

\begin{equation}
\label{eq:OptP22}
\begin{aligned}
& \underset{P,P_p}{\text{minimize}}
& & \frac{Kc + \frac{\sigma_d^2}{\alpha^2 P}}{\hat{e}c + \check{e}ca} - K 
~~~~\text{subject to}
~~P \tau_d + P_p \tau_p = P_{\text{tot}}.
\end{aligned}
\end{equation}

By substituting $P = (P_\textup{tot} - P_p\tau_p)/\tau_d$, the values of $\hat{e}$ and $\check{e}$ from \eqref{eq:es} into the objective function, the optimization task in \eqref{eq:OptP22} is further equivalent with:
\begin{equation}
\label{eq:OptP6}
\begin{aligned}
& \underset{P_p}{\text{minimize}}
& & \frac{ \left(Kc + \frac{\sigma_d^2 \tau_d}{\alpha^2 (P_{\textup{tot}} - P_p \tau_p)} \right)\left( \left(c+ \frac{\sigma_p^2}
{\alpha^2 P_p \tau_p} \right)^2 + a^2 c^2 \right)}{ (a^2+1) \frac{\sigma_p^2}{\alpha^2 P_p \tau_p} + c - a^2 c}.
\end{aligned}
\end{equation}

\rev{
Notice that this expression approaches infinity both when $P_p \rightarrow 0$ and when $P_p \rightarrow P_\textup{tot}/\tau_p$:
\begin{align}
&\lim_{P_p \rightarrow 0} \frac{ \left(Kc + \frac{\sigma_d^2 \tau_d}{\alpha^2 (P_{\textup{tot}} - P_p \tau_p)} \right)\left( \left(c+ \frac{\sigma_p^2}
{\alpha^2 P_p \tau_p} \right)^2 + a^2 c^2 \right)}{ (a^2+1) \frac{\sigma_p^2}{\alpha^2 P_p \tau_p} + c - a^2 c}
\nonumber \\&
=
\left(Kc + \frac{\sigma_d^2 \tau_d}{\alpha^2 P_{\textup{tot}}} \right) \lim_{P_p \rightarrow 0} \frac{ \left(c+ \frac{\sigma_p^2}
{\alpha^2 P_p \tau_p} \right)^2 + a^2 c^2}{ (a^2+1) \frac{\sigma_p^2}{\alpha^2 P_p \tau_p} + c - a^2 c} \times \frac{P_p}{P_p}
\nonumber \\&
=
\left(Kc + \frac{\sigma_d^2 \tau_d}{\alpha^2 P_{\textup{tot}}} \right) \frac{ \lim\limits_{P_p \rightarrow 0} P_p\left(c+ \frac{\sigma_p^2}
{\alpha^2 P_p \tau_p} \right)^2}{ (a^2+1) \frac{\sigma_p^2}{\alpha^2 \tau_p}} = \infty;
\end{align}
}
\rev{
\begin{align}
&\lim_{P_p \rightarrow P_\textup{tot}/\tau_p} \frac{ \left(Kc + \frac{\sigma_d^2 \tau_d}{\alpha^2 (P_{\textup{tot}} - P_p \tau_p)} \right)\left( \left(c+ \frac{\sigma_p^2}
{\alpha^2 P_p \tau_p} \right)^2 + a^2 c^2 \right)}{ (a^2+1) \frac{\sigma_p^2}{\alpha^2 P_p \tau_p} + c - a^2 c}
\nonumber \\&
=
\frac{ \left(c+ \frac{\sigma_p^2}
{\alpha^2 P_\textup{tot}} \right)^2 + a^2 c^2}{ (a^2+1) \frac{\sigma_p^2}{\alpha^2 P_\textup{tot}} + c - a^2 c}
\lim_{P_p \rightarrow P_\textup{tot}/\tau_p}
\left(Kc + \frac{\sigma_d^2 \tau_d}{\alpha^2 (P_{\textup{tot}} - P_p \tau_p)} \right) \nonumber \\
&~~~=
\infty.
\end{align}
Since the expression to minimize is positive over the interval $\left( 0, P_\textup{tot}/\tau_p \right)$,
and approaches infinity at the edges of the interval, there is a global minimum in the interval which is also a stationary point.
To find the set of all stationary points, we calculate
the derivative of the expression in equation \eqref{eq:OptP6} with respect to $P_p$.
We have:
\begin{align}
&\frac{d}{dP_p} \left(Kc + \frac{\sigma_d^2 \tau_d}{\alpha^2 (P_{\textup{tot}} - P_p \tau_p)} \right) =
\frac{\sigma_d^2 \tau_d}{\alpha^2 (P_{\textup{tot}} - P_p \tau_p)^2} \nonumber \\
&\frac{d}{dP_p} \left( \left(c+ \frac{\sigma_p^2} {\alpha^2 P_p \tau_p} \right)^2 + a^2 c^2 \right) = \nonumber \\
&~~~~~~~~~~~~~~~~~~~~~~~~~~~~~~=-2\left(c+ \frac{\sigma_p^2} {\alpha^2 P_p \tau_p} \right)\frac{\sigma_p^2} {\alpha^2 P_p^2 \tau_p} \nonumber \\
&\frac{d}{dP_p} \left((a^2+1) \frac{\sigma_p^2}{\alpha^2 P_p \tau_p} + c - a^2 c \right) = -(a^2+1) \frac{\sigma_p^2}{\alpha^2 P_p^2 \tau_p}.
\end{align}
From this we can calculate the derivative of \eqref{eq:OptP6} with respect to $P_p$,
which is a rational function with numerator equal to the polynomial given in equation \eqref{eq:OptP3},
\qq{Hence,} this polynomial has at least one positive root in the interval $\left( 0, P_\textup{tot}/\tau_p \right)$,
one of which gives the solution to the optimization task \eqref{eq:OptP6}, and hence the optimal pilot power.
}
\end{proof}

\bibliography{MSEGame_gf,MMSERef}
\end{document}